\documentclass[11pt]{article}
\usepackage[margin=2.54cm]{geometry}                
\usepackage{graphicx}
\usepackage{amssymb}
\usepackage{epstopdf}
\usepackage{amsmath}
\usepackage{amsthm}

\newcommand{\rd}{{\rm d}}

\newtheorem{thm}{Theorem}

\newtheorem{corollary}[thm]{Corollary}
\newtheorem{lemma}[thm]{Lemma}
\newtheorem{prop}[thm]{Proposition}
\newtheorem{example}[thm]{Example}
\newtheorem{remark}[thm]{Remark}

\usepackage{amsfonts} 
\usepackage{fullpage} 

\begin{document}

\renewcommand{\thefootnote}{\fnsymbol{footnote}}
\title{Lieb-Robinson bounds for classical anharmonic lattice systems}
 
\author{Hillel Raz\\[10pt]
Department of Mathematics\\ University of California at Davis\\
Davis CA 95616, USA\\
\vspace{.3cm} Email: hraz@math.ucdavis.edu\\
\vspace{.3cm} and\\
\vspace{.3cm} Robert Sims \\
Department of Mathematics\\ 
University of Arizona\\
Tucson, AZ 85721, USA\\
Email:  rsims@math.arizona.edu}

\date{Version: \today }
\maketitle
\bigskip

\abstract{We prove locality estimates, in the form of Lieb-Robinson bounds, for classical oscillator systems defined on 
a lattice. Our results hold for the harmonic system and a variety of anharmonic perturbations with long range interactions. 
The anharmonic estimates are applicable to a special class of observables, the Weyl functions, and the bounds which 
follow are not only independent of the volume but also the initial condition. }

\footnotetext[1]{Copyright \copyright\ 2009 by the authors. 
This paper may be reproduced, in its entirety, for non-commercial
purposes.}

\setcounter{section}{0}

%
%
%
%

\section{Introduction} \label{sec:intro}

A notion of locality is crucial in rigorously analyzing most physical systems. 
Typically, sets of local observables are associated with bounded regions of space, 
and one is interested in how these observables evolve dynamically with respect to 
the interactions governing the system. In relativistic theories, the evolution 
of a local observable remains local, i.e. the associated dynamics is restricted to a 
light cone.  For non-relativistic models, such as those we will be considering in the 
present work, the dynamics does not preserve locality in the sense that, generically, 
an observable initially chosen localized to a particular site is immediately evolved 
into an observable dependent on all sites of the system. 

In 1972,  Lieb and Robinson \cite{lieb1972} explored a quasi-locality of the dynamics corresponding to
non-relativistic quantum spin systems. Roughly speaking, a quantum spin system is described by a self-adjoint 
Hamiltonian, which describes the inteactions of the system, and its associated Heisenberg dynamics,
see e.g. \cite{bratteli1997} for more details. The estimates they proved, which we will refer to as Lieb-Robinson bounds, demonstrate that, 
up to exponentially small errors, the time evolution of a local observable remains essentially supported in a
ball of radius proportional to $v|t|$ for some $v>0$. This quantity $v$ defines a natural velocity of propagation, 
and it can be estimated in terms of the system's free parameters, for example, the interaction strength
of the Hamiltonian. 

The models analyzed in this paper will correspond to a classical system of oscillators evolving 
according to a Hamiltonian dynamics. Hamiltonians of this type have frequently appeared
in the literature as their analysis provides an important means of 
studying the emergence of non-equilibrium phenomena in macroscopic systems.
For example, rigorous results on the existence of the thermodynamic limit for 
these models date back to \cite{LLL}.  A notion of quasi-locality, similar to the Lieb-Robinson bounds 
mentioned above, for these oscillator systems was originally considered in 1978 by Marchioro et. al. in \cite{MPPT}, and 
a recent generalization of these estimates appeared in \cite{butta2007}.  Both of these results were obtained in the spirit of deriving 
an analogue of the Lieb-Robinson bounds found in \cite{lieb1972}.

Over the past few years a number of important improvements on the original Lieb-Robinson bounds have appeared in the 
literature \cite{hastings2004,nachtergaele2005,hastings2005, NaOgSi,harm}, see \cite{review} for the most current review article. 
These new estimates have found a variety of intriguing applications \cite{hastings2004, bravyi2006, nasi, hastings2007}, but perhaps
most interestingly for the present work, the results found in \cite{harm} establish bounds which are applicable 
beyond the context of quantum spin systems. In \cite{harm}, the authors prove a version of the Lieb-Robinson 
bounds for quantum anharmonic lattice systems. Motivated by the methodology introduced in \cite{harm}, we
are returning to the classical setting to re-derive distinct estimates for anharmonic lattice systems. 

To express our locality results more precisely, we introduce the following notation.
We will consider systems confined to a large but finite subset $\Lambda \subset \mathbb{Z}^{\nu}$;
here $\nu \geq 1$ is an integer. With each site $x \in \Lambda$, we will associate an oscillator with 
coordinate $q_x \in \mathbb{R}$ and momentum $p_x \in \mathbb{R}$. The state of the system in
$\Lambda$ will be described by a sequence $\mathrm{x} = \{ (q_x, p_x) \}_{x \in \Lambda}$, and
phase space, i.e. the set of all such sequences, will be denoted by $\mathcal{X}_{\Lambda}$.

A Hamiltonian, $H$, is a real-valued function on phase space. Typically the Hamiltonian of interest
generates a flow, $\Phi_t$, on phase space. Specifically, given $H: \mathcal{X}_{\Lambda} \to \mathbb{R}$
one defines, for any $t \in \mathbb{R}$, a function $\Phi_t : \mathcal{X}_{\Lambda} \to \mathcal{X}_{\Lambda}$ by 
setting $\Phi_t( \mathrm{x}) = \{ (q_x(t), p_x(t)) \}_{x \in \Lambda}$, the sequence whose components satisfy Hamilton's equations:
for any $x \in \Lambda$,
\begin{equation}
\begin{split}
\dot{q}_x(t) \, = \, \frac{ \partial H}{ \partial p_x} \left( \Phi_t(\mathrm{x}) \right), \\
\dot{p}_x(t) \, = \, -\frac{ \partial H}{ \partial q_x} \left( \Phi_t(\mathrm{x}) \right),
\end{split} 
\end{equation}
with initial condition $\Phi_0(\mathrm{x}) = \mathrm{x}$.

To measure the effects of this Hamiltonian dynamics on the system, one introduces observables.
An observable $A$ is a complex-valued function of phase space. We will denote by
$\mathcal{A}_{\Lambda}$ the space of all local observables in $\Lambda$, i.e. 
the set of all functions $A: \mathcal{X}_{\Lambda} \to \mathbb{C}$. A given Hamiltonian, 
$H$, generates a dynamics $\alpha_t$ on the space of local observables in the sense that, for any $t \in \mathbb{R}$,
the dynamics $\alpha_t : \mathcal{A}_{\Lambda} \to \mathcal{A}_{\Lambda}$
is defined by setting $\alpha_t(A) = A \circ \Phi_t$. 

For the locality result we will present, the notion of support of a local observable is important. Given $A \in \mathcal{A}_{\Lambda}$, the
support of $A$ is defined to be the minimal set $X \subset \Lambda$ for which $A$ depends only on those parameters $q_x$ or $p_x$ with $x \in X$.

As in \cite{MPPT}, see also \cite{butta2007}, our locality result will be expressed in terms of the Poisson bracket between local observables. Here the Poisson bracket
is the observable given by
\begin{equation}
\left\{ A, B \right\} \, = \, \sum_{x \in \Lambda} \frac{ \partial A}{ \partial q_x} \cdot \frac{ \partial B}{ \partial p_x} \, - \, \frac{ \partial A}{ \partial p_x} \cdot \frac{ \partial B}{ \partial q_x} \, ,
\end{equation} 
for sufficiently smooth observables $A$ and $B$.

Observe that for disjoint subsets $X,Y \subset \Lambda$ and observables $A$ with support in $X$ and $B$ with support in $Y$, it is clear that
$\{ A, B \} = 0$. The quasi-locality question of interest in this context is: given a Hamiltonian $H$, its corresponding dynamics $\alpha_t$, and a pair of
observables $A$ and $B$ with disjoint supports, is there a bound on the quantity $\{ \alpha_t(A), B \}$ for small times $t$? Physically, one expects that
if the Hamiltonian is comprised of local interaction terms, then there should be a bound on the velocity of propagation through the system. Such intuition
could be confirmed by establishing an estimate of the form
\begin{equation} \label{eq:clrb}
\left|  \left\{ \alpha_t(A), B \right\} ( \mathrm{x})  \right| \, \leq \, C e^{- \mu( d(X,Y) - v |t|)},
\end{equation} 
where $d(X,Y)$ denotes the distance between the supports of the local observables $A$ and $B$. This bound demonstrates that for times $t$ with
$|t| \leq d(X,Y) /v$, the Poisson bracket remains exponentially small, and the number $v>0$ appearing above is a bound on the system's velocity. 
In proving estimates of the form (\ref{eq:clrb}), special attention must be given to the dependence of the
constants $C$, $\mu$, and $v$ on the observables $A$ and $B$, the initial condition $\mathrm{x}$, and the free parameters in the Hamiltonian.
Most crucially, these constants must be independent of the underlying volume $\Lambda$, so that they persist in the thermodynamic limit; once
the existence of such a limit has been established. 

As we have mentioned before, bounds of the form (\ref{eq:clrb}) have appeared in the literature, see \cite{MPPT} and more 
recently \cite{butta2007}, for a variety of different Hamiltonians. Both our approach and our 
estimates are distinct from those mentioned above. For example, we do not work with time invariant states, and
our bounds are independent of the initial conditions. Our main goal is to provide a new method for establishing 
these estimates, and we are strongly motivated by the quantum techniques found in \cite{harm}.  

We begin by considering finite volume restrictions of the harmonic Hamiltonian, i.e. $H_h^{\Lambda} : \mathcal{X}_{\Lambda} \to \mathbb{R}$ is
given by
\begin{equation}
H_h^{\Lambda}( \mathrm{x}) \, = \, \sum_{x \in \Lambda} p^2_x + \omega^2 q^2_x + \sum_{j = 1}^{\nu} \lambda_j (q_x - q_{x+e_j})^2,
\end{equation} 
 where $e_j$, for $j = 1, \ldots, \nu$, are the cannonical basis vectors in $\mathbb{Z}^{\nu}$, and the parameters $\omega \geq 0$ and
 $\lambda_j \geq 0$ are the on-site and coupling strength, respectively.  As is well-known, a variety of explicit calculations may be performed
 for this harmonic Hamiltonian. Perhaps most importantly, for any integer $L \geq 1$ and each subset $\Lambda_L = (-L, L]^{\nu} \subset \mathbb{Z}^{\nu}$, 
 the flow $\Phi_t^{h,L}$ corresponding to $H_h^{\Lambda_L}$ may be explicitly computed, see Section~\ref{subsec:sb} for details. 
 Once this is known, a locality estimate easily follows for a specific set of observables. We will equip the set of local observables
 $\mathcal{A}_{\Lambda_L}$ with the sup-norm, and we will say that $A \in \mathcal{A}_{\Lambda_L}$ is bounded if
 \begin{equation}
 \| A \|_{\infty} = \sup_{\mathrm{x} \in \mathcal{X}_{\Lambda_L}} | A( \mathrm{x}) | 
 \end{equation}
 is finite. Furthermore, we will denote by $\mathcal{A}_{\Lambda_L}^{(1)}$ the set of all $A \in \mathcal{A}_{\Lambda_L}$ for which:
given any $x \in \Lambda_L$, $\frac{\partial A}{\partial q_x} \in \mathcal{A}_{\Lambda_L}$, 
$\frac{\partial A}{\partial p_x} \in \mathcal{A}_{\Lambda_L}$, and
\begin{equation}
\| \partial A \|_{\infty} = \sup_{x \in \Lambda_L} \max \left( \left\| \frac{\partial A}{\partial q_x} \right\|_{\infty}, \left\| \frac{\partial A}{\partial p_x} \right\|_{\infty} \right) \, < \, \infty \, .
\end{equation}
We can now state our first result.
 \begin{thm}\label{thm:harm}
Let $X$ and $Y$ be finite subsets of $\mathbb{Z}^{\nu}$ and take $L_0$ to be the smallest integer such that $X, Y \subset \Lambda_{L_0}$.
For any $L \geq L_0$, denote by $\alpha_t^{h,L}$ the dynamics corresponding to $H_h^{\Lambda_L}$. 
For any $\mu >0$ and any observables $A,B \in \mathcal{A}_{\Lambda_{L_0}}^{(1)}$ with 
support of $A$ in $X$ and support of $B$ in $Y$, there exist positive numbers $C$ and $v_h$, both independent of $L$, such that the bound 
\begin{equation} \label{eq:thm1}
\left\| \left\{ \alpha_t^{h, L}(A), B \right\} \right\|_{\infty} \le C \| \partial A\|_{\infty} \| \partial B \|_{\infty} \min ( |X|, |Y|) e^{-\mu \left(d(X,Y) - v_{\rm{h}} |t| \right)}
\end{equation}
holds for all $t \in \mathbb{R}$.
\end{thm}

Some additional comments are in order.
First, the quantity $d(X,Y)$ appearing above denotes the distance between the sets $X$ and $Y$, measured in the
$L^1$-sense, and for any $Z \subset \Lambda_L$, the number $|Z|$ is the cardinality of $Z$. Next, the fact that the bound (\ref{eq:thm1}) is true for any $\mu >0$ implies that the Poisson bracket above has arbitrarily fast
exponential decay in space. To achieve faster decay in space, however, the numbers $C$ and $v_h$ increase. We describe an
optimal harmonic velocity $v_h( \mu)$ in Section \ref{subsec:genloc}. 

One novelty of our approach is that the bound in (\ref{eq:thm1}) is not only independent of the length scale $L$, it is also
independent of the initial condition $\mathrm{x} \in \mathcal{X}_{\Lambda_L}$. This fact remains true when we consider
anharmonic perturbations, see Theorem \ref{thm:anharm} below, and thereby distinguishes our result from that of \cite{MPPT} and \cite{butta2007}.

Our next result, Theorem~\ref{thm:anharm} below, concerns single site perturbations of the harmonic Hamiltonian. 
To state this precisely, fix a function $V: \mathbb{R} \to \mathbb{R}$. For any site 
$z \in \mathbb{Z}^{\nu}$ define $V_z : \mathcal{X}_{\Lambda_L} \to \mathbb{R}$ by setting
$V_z( \mathrm{x}) = V(q_z)$. We consider finite volume anharmonic Hamiltonians
$H^{\Lambda_L} : \mathcal{X}_{\Lambda_L} \to \mathbb{R}$ of the form 
\begin{equation}
H^{\Lambda_L} = H_h^{\Lambda_L}  + \sum_{z \in \Lambda_L} V_z.  
\end{equation}
In order to prove Theorem~\ref{thm:anharm}, we need the following assumptions on $V$:
$V \in C^2(\mathbb{R})$, $V' \in L^1(\mathbb{R})$, $V'' \in L^{\infty}(\mathbb{R})$, and
$$
\kappa_V = \int\,  |r| \, \left| \widehat{V'} (r) \right| \, \rd r \, < \infty \, .
$$
Here $\widehat{V'}$ is the Fourier transform of $V'$. Under these assumptions, we prove a locality 
result analogous to Theorem \ref{thm:harm}. As is discussed in Section~\ref{subsec:weyl}, see also
the proof in Section~\ref{sec:anharmss}, a specific class of observables, the Weyl functions, are 
particularly well-suited for our considerations, and they are defined as follows. For any function
 $f : \Lambda_L \to \mathbb{C}$, the Weyl function generated by $f$, denoted by $W(f)$, is the observable $W(f) : \mathcal{X}_{\Lambda_L} \to \mathbb{C}$ given by
 \begin{equation} \label{eq:defweyl}
 [W(f)]( \mathrm{x}) \, = \, \mbox{exp} \left[ \, i \, \sum_{x \in \Lambda_L} \mbox{Re}[f(x)] q_x + \mbox{Im}[f(x)] p_x \, \right] \, .
 \end{equation}
 Clearly, if $f$ is supported in $X \subset \Lambda_L$, then $W(f)$ is supported in $X$ as well. 
 Moreover, it is easy to see that for any function $f: \Lambda_L \to \mathbb{C}$, $\| W(f) \|_{\infty} = 1$.
 Our next result is
 \begin{thm}\label{thm:anharm}
Let $V : \mathbb{R} \to \mathbb{R}$ satisfy $V \in C^2(\mathbb{R})$, $V' \in L^1(\mathbb{R})$, $V'' \in L^{\infty}(\mathbb{R})$, and $\kappa_V$, as defined above,
is finite. Take $X$ and $Y$ to be finite subsets of $\mathbb{Z}^{\nu}$ and let $L_0$ to be the smallest integer such that $X, Y \subset \Lambda_{L_0}$.
For any $L \geq L_0$, denote by $\alpha_t^L$ the dynamics corresponding to $H^{\Lambda_L}$. 
For any $\mu >0$ and any functions $f,g: \Lambda_{L_0} \to \mathbb{C}$ with 
support of $f$ in $X$ and support of $g$ in $Y$, there exist positive numbers $C$ and $v_{ah}$, both independent of $L$, such that the bound 
\begin{equation} \label{eq:thm2}
\left\| \left\{ \alpha_t^L(W(f)), W(g) \right\} \right\|_{\infty} \le C \|f\|_{\infty} \|g\|_{\infty} \min(|X|,|Y|) e^{-\mu \left(d(X,Y) - v_{ah} |t| \right)}
\end{equation}
holds for all $t \in \mathbb{R}$.
\end{thm}
The assumptions on $V$ above are sufficient to imply that $V$ is bounded. For this reason, our results
do not apply to more substantial perturbations, e.g., those of the form $V_z( \mathrm{x}) = q_z^4$. 
In \cite{butta2007}, the authors do consider, for example, quartic perturbations. They prove that, for reasonable time invariant states and a set of initial conditions of full measure,
after a time $t$ local perturbations of thermal equilibrium are exponentially small in $\log^2(t)$ at a distance larger
that $t \log^{\alpha}(t)$. This is insufficient to conclude the existence of a finite velocity $v>0$ as we have 
discussed above. It is an interesting question to determine whether or not this genuinely describes the 
behavior of such systems. We will not answer this question in the present work. 

The paper is organized as follows. In Section~\ref{sec:harm}, we discuss our results concerning the
Harmonic Hamiltonian and prove Theorem~\ref{thm:harm}. Using an interpolation argument, we prove
Theorem~\ref{thm:anharm} in Section~\ref{sec:anharmss}. This result demonstrates that our locality bounds
for Weyl functions are preserved under certain single-site anharmonic perturbations. In Section~\ref{sec:anharmms}, we 
generalize Theorem~\ref{thm:anharm} to cover a wide class of multi-site perturbations. Finally,
Section~\ref{sec:app} contains a variety of useful solution estimates used throughout the paper.

%
%
%
%

\setcounter{equation}{0}

\section{The Harmonic Hamiltonian} \label{sec:harm}
The main goal of this section is to prove Theorem~\ref{thm:harm}.
For the convenience of the reader, we begin with a subsection
describing some basic features of the harmonic Hamiltonian.
In particular, we reintroduce the Hamiltonian and find an explicit 
expression for the corresponding flow. In the subsections which 
follow, we prove two locality estimates.
The first is valid for a general class of smooth and bounded observables.
The next holds for a special class of observables, the Weyl functions. 
This latter result will be particularly useful in subsequent sections. 

\subsection{Some basics} \label{subsec:sb}
For any integer $L \geq 1$, we consider subsets  $\Lambda_L = (-L, L]^{\nu} \subset \mathbb{Z}^{\nu}$ and
the finite volume harmonic Hamiltonian $H_h^{\Lambda_L}: \mathcal{X}_{\Lambda_L} \to \mathbb{R}$ given by
\begin{equation}\label{eq:HarmHam}
H_h^{\Lambda_L}( \mathrm{x}) = \sum_{x \in \Lambda_L} p_x^2 + \omega^2 q_x^2 +\sum_{j=1}^{\nu} \lambda_j (q_x - q_{x+e_j})^2.
\end{equation}
Here, for each $j = 1, \ldots, \nu$, the $e_j$ are the canonical basis vectors in $\mathbb{Z}^{\nu}$, $\omega  \geq 0$, and $\lambda_j  \geq 0$.
The model in (\ref{eq:HarmHam}) is defined with periodic boundary conditions, in the sense that $q_{x + e_j} = q_{x-(2L-1)e_j}$ if $x \in \Lambda_L$ 
but $x + e_j \notin \Lambda_L$. 

Our first task is to provide an explicit expression for the flow corresponding to (\ref{eq:HarmHam}).
In doing so, we will fix an integer value of $L \geq 1$ and drop its dependence in a variety of
quantities to ease notation. Given any $\mathrm{x} \in \mathcal{X}_{\Lambda_L}$ and $t \in \mathbb{R}$, the components 
of $\Phi_t^h( \mathrm{x}) = \{ \left( q_x(t), p_x(t) \right) \}_{x \in \Lambda_L}$ satisfy the following coupled system of differential equations:
for each $x \in \Lambda_L$ and $t \in \mathbb{R}$,
\begin{eqnarray}\label{eq:Hamiltons}
\nonumber \dot{q}_x(t) &=&  2p_x(t) \\  
\dot{p}_x(t) &=& -2\omega^2q_x(t) - 2\sum_{j =1}^{\nu} \lambda_j \left(2q_x(t) - q_{x+e_j}(t) - q_{x - e_j}(t) \right)
\end{eqnarray}
with initial condition $\{(q_x(0), p_x(0)) \}_{x \in \Lambda_L} = \mathrm{x}$.
Introducing Fourier variables, the system defined by (\ref{eq:Hamiltons}) decouples which leads to
an exact solution. This is the content of Lemma~\ref{lem:harmsol} found below.

Before stating Lemma~\ref{lem:harmsol}, it is useful to introduce some additional notation. 
Fourier sums will be defined via the set $\Lambda_L^*$ given by
$$
\Lambda_L^* = \left\{ \, \frac{x \pi}{L} \, : \, x \in \Lambda_L \, \right\} \, .
$$
Note that $\Lambda_L^* \subset (- \pi, \pi]^{\nu}$ and $| \Lambda_L^*| = | \Lambda_L| = (2L)^{\nu}$.
The following functions play an important role in our calculations. 
Suppose $\omega >0$ and take
$\gamma : \Lambda_L^* \to \mathbb{R}$ to be given by
\begin{equation} \label{eq:gamma}
\gamma(k) = \sqrt{\omega^2 + 4\sum_{j =1}^{\nu}\lambda_j \sin^2(k_j/2)} \, ,
\end{equation}
and for each $m \in \{ -1, 0, 1 \}$ and any $t \in \mathbb{R}$, set $h^{(m)}_t: \Lambda_L \to \mathbb{R}$
to be
\begin{equation}\label{eq:h}
\begin{split}
h^{(-1)}_t(x) &= {\rm Im} \left[\frac{1}{|\Lambda_L|}\sum_{k \in \Lambda_L^*}\frac{e^{i(k \cdot x-2\gamma(k)t)}}{\gamma(k)} \right],
\\
h^{(0)}_t(x) &={\rm Re}\left[\frac{1}{|\Lambda_L|}\sum_{k \in \Lambda_L^*}e^{i(k \cdot x - 2\gamma(k)t)}\right],
\\
h^{(1)}_t(x) &= {\rm Im}\left[ \frac{1}{|\Lambda_L|}\sum_{k \in \Lambda_L^*}\gamma(k) \, e^{i(k \cdot x-2\gamma(k)t)}\right] .
\end{split}
\end{equation}
Each of these functions depend on the length scale $L$, however, we are suppressing that dependence. 

\begin{lemma} \label{lem:harmsol}
Suppose $\omega >0$. For any $\mathrm{x} \in \mathcal{X}_{\Lambda_L}$ and $t \in \mathbb{R}$, the mapping $\Phi_t^h : \mathcal{X}_{\Lambda_L} \to \mathcal{X}_{\Lambda_L}$
is well-defined. In particular, for each $x \in \Lambda_L$ and $t \in \mathbb{R}$, the components of $\Phi_t^h( \mathrm{x}) = \left\{ (q_x(t), p_x(t) )\right\}_{x \in \Lambda_L}$ are
given by
\begin{equation}\label{eq:q_x(t)}
q_x(t) = \sum_{y \in \Lambda_L} q_y(0) \, h^{(0)}_t(x-y) - p_y(0) \, h^{(-1)}_t(x-y)
\end{equation}
and
\begin{equation}\label{eq:p_x(t)}
p_x(t) = \sum_{y \in \Lambda_L} q_y(0) \, h^{(1)}_t(x-y) + p_y(0) \, h^{(0)}_t(x-y).
\end{equation}
Here, if necessary, the function values $h^{(m)}_t(x-y)$ are defined by periodic extension, and we regard 
$\mathrm{x} = \{ (q_x(0), p_x(0) ) \}_{x \in \Lambda_L}$.
\end{lemma}

\begin{proof}
Taking a second derivative of (\ref{eq:Hamiltons}), we find that for each $x \in \Lambda_L$ and any $t \in \mathbb{R}$,
\begin{equation}\label{eq:2dqp}
\begin{split}
\ddot{q}_x(t) &= -4\omega^2 q_x(t) - 4\sum_{j=1}^{\nu} \lambda_j \left( 2q_x(t) - q_{x+e_j}(t) - q_{x-e_j}(t) \right)
\\
\ddot{p}_x(t) &= -4\omega^2 p_x(t) - 4\sum_{j=1}^{\nu} \lambda_j \left( 2p_x(t) - p_{x+e_j}(t) - p_{x-e_j}(t) \right).
\end{split}
\end{equation}
For any $k \in \Lambda_L^*$ and $t \in \mathbb{R}$, set  
\begin{equation} \label{eq:Q+Pk}
Q_k(t) \, = \, \frac{1}{ \sqrt{ | \Lambda_L |}} \sum_{x \in \Lambda_L} e^{- i k \cdot x} q_x(t) \quad \mbox{and} \quad
P_k(t) \, = \, \frac{1}{ \sqrt{ | \Lambda_L |}} \sum_{x \in \Lambda_L} e^{- i k \cdot x} p_x(t). 
\end{equation}
Inserting (\ref{eq:2dqp}) into the second derivative of (\ref{eq:Q+Pk}),  we find an equivalent system of
uncoupled differential equations. In fact, for each $k \in \Lambda_L^*$ and any $t \in \mathbb{R}$,
\begin{equation}\label{eq:2dQP}
\begin{split}
\ddot{Q}_k(t) &= -4\omega^2 Q_k(t) - 4\sum_{j=1}^{\nu} \lambda_j \left( 2 - e^{ik_j} - e^{-ik_j} \right) Q_k(t) = - 4 \gamma(k)^2 Q_k(t)
\\
\ddot{P}_k(t) &= -4\omega^2 P_k(t) - 4\sum_{j=1}^{\nu} \lambda_j \left( 2 - e^{ik_j} - e^{-ik_j} \right)P_k(t) = - 4 \gamma(k)^2P_k(t) \, ,
\end{split}
\end{equation}
where $\gamma$ is as in (\ref{eq:gamma}). The solution of (\ref{eq:2dQP}) is given by
\begin{equation}\label{eq:QPsol}
\begin{split}
Q_k(t) &=  C_k e^{-2i\gamma(k)t} + \overline{C_{-k}} e^{2i\gamma(k)t} \\
P_k(t) &=  D_k e^{-2i\gamma(k)t} + \overline{D_{-k}} e^{2i\gamma(k)t},
\end{split}
\end{equation}
where  $-k$ is defined to be the element of $\Lambda_L^*$ whose components are given by
$$
 (-k)_j = \left\{ \begin{array}{cc} -k_j, & \text{if } |k_j|< \pi, \\ \pi , & \text{otherwise.} \end{array} \right.
$$
The relationship between the coefficients in (\ref{eq:QPsol}) above is derived using the fact that
the initial condition is real-valued, e.g., 
$$
Q_k(0) = \overline{Q_{-k}(0)} \quad \mbox{and} \quad \dot{Q}_k(0) = \overline{\dot{Q}_{-k}(0)}.
$$
Using Fourier inversion, we recover the components of the flow from (\ref{eq:QPsol}). In fact,
\begin{eqnarray} \label{eq:qsol}
q_x(t) & = & \frac{1}{ \sqrt{ | \Lambda_L |}} \sum_{k \in \Lambda_L^*} e^{i k \cdot x} Q_k(t) \nonumber \\
& = &  \frac{1}{ \sqrt{ | \Lambda_L |}} \sum_{k \in \Lambda_L^*} C_k e^{i \left( k \cdot x - 2\gamma(k)t \right)} +
\overline{C_k} e^{-i \left( k \cdot x - 2\gamma(k)t \right)} \, , 
\end{eqnarray}
and similarly, we find that
\begin{equation}
p_x(t) \, = \, \frac{1}{ \sqrt{ | \Lambda_L |}} \sum_{k \in \Lambda_L^*} D_k e^{i \left( k \cdot x - 2\gamma(k)t \right)} +
\overline{D_k} e^{-i \left( k \cdot x - 2\gamma(k)t \right)}. 
\end{equation}

To express these solutions explicitly in terms of the initial condition, we observe that
\begin{equation}\label{eq:Q_k(0)}
Q_k(0) = C_k + \overline{C_{-k}} \quad \mbox{and} \quad P_k(0) = D_k + \overline{D_{-k}} \, ,
\end{equation}
and introduce  
\begin{equation} \label{eq:beqns}
B_k \, = \, \frac{1}{ \sqrt{2 \gamma(k)}} \, P_k(0) - i \sqrt{ \frac{\gamma(k)}{2}} \, Q_k(0) \quad {\rm with} \quad
\overline{B_k} \, = \, \frac{1}{ \sqrt{2 \gamma(k)}} \, P_{-k}(0) + i \sqrt{ \frac{\gamma(k)}{2}} \, Q_{-k}(0) \, . 
\end{equation}
It is easy to see that
\begin{equation} \label{eq:Q+PBk1}
Q_k(0) \, = \, \frac{i}{ \sqrt{ 2 \gamma(k)}} \left( B_k \, - \, \overline{B_{-k}} \right) \quad \mbox{and} \quad
P_k(0) \, = \, \sqrt{ \frac{\gamma(k)}{2}} \left( B_k \, + \, \overline{B_{-k}} \right) \, ,
\end{equation}
and therefore,
\begin{equation}\label{eq:C_k=b}
C_k = \frac{i B_k}{\sqrt{2\gamma(k)}} \quad \mbox{and} \quad D_k = \sqrt{\frac{\gamma(k)}{2}} B_k \, .
\end{equation}
Plugging this into (\ref{eq:qsol}), we find that
\begin{eqnarray}\label{eq:qsol2}
q_x(t) &=& \frac{1}{\sqrt{|\Lambda_L|}} \sum_{k \in \Lambda_L^*} \frac{i B_k}{\sqrt{2\gamma(k)}} e^{i(k \cdot x - 2\gamma(k)t)} - \frac{i \overline{B_k}}{\sqrt{2\gamma(k)}}e^{-i(k \cdot x - 2\gamma(k)t)}
  \\
& = & \frac{1}{2 \sqrt{|\Lambda_L|}}\sum_{k\in\Lambda_L^*} Q_k(0) e^{i(k \cdot x - 2\gamma(k)t)} + \overline{Q_k(0)} e^{-i(k \cdot x - 2\gamma(k)t)} 
\nonumber \\
& \mbox{ } & \quad + \frac{i}{2 \sqrt{|\Lambda_L|}}\sum_{k\in\Lambda_L^*} \frac{P_k(0)}{\gamma(k)} e^{i(k \cdot x - 2\gamma(k)t)} - 
\frac{ \overline{P_k(0)}}{\gamma(k)} e^{-i(k \cdot x - 2\gamma(k)t)} 
\nonumber \\ 
& = &  \frac{1}{ \sqrt{|\Lambda_L|}}  \sum_{k\in\Lambda_L^*} {\rm Re} \left[ Q_k(0) e^{i(k \cdot x - 2\gamma(k)t)} \right] -  
\frac{1}{\sqrt{|\Lambda_L|}}\sum_{k\in\Lambda_L^*} {\rm Im} \left[ \frac{P_k(0)}{\gamma(k)} e^{i(k \cdot x - 2\gamma(k)t)} \right] \, . \nonumber
\end{eqnarray}
Moreover, one finds that
\begin{equation}
{\rm Re} \left[ Q_k(0) e^{i(k \cdot x - 2\gamma(k)t)} \right] = \frac{1}{ \sqrt{| \Lambda_L|}} \sum_{y \in \Lambda_L} q_y(0) 
{\rm Re} \left[ e^{i(k \cdot (x-y) - 2\gamma(k)t)} \right]
\end{equation}
while
\begin{equation}
{\rm Im} \left[ \frac{P_k(0)}{\gamma(k)} e^{i(k \cdot x - 2\gamma(k)t)} \right] = \frac{1}{ \sqrt{| \Lambda_L|}} \sum_{y \in \Lambda_L} p_y(0) 
{\rm Im} \left[ \frac{1}{\gamma(k)} e^{i(k \cdot (x-y) - 2\gamma(k)t)} \right] \, .
\end{equation}
With the functions $h_t^{(m)}$, as defined in (\ref{eq:h}), we conclude that
\begin{equation}
q_x(t) = \sum_{y \in \Lambda_L} q_y(0) h_t^{(0)}(x-y) - p_y(0) h_t^{(-1)}(x-y),
\end{equation}
as claimed in (\ref{eq:q_x(t)}). A similar calculation yields (\ref{eq:p_x(t)}).
Since the functions $h_t^{(m)}$ are real valued, so too are the solutions 
$q_x(t)$ and $p_x(t)$. This proves Lemma~\ref{lem:harmsol}.
\end{proof}

\begin{remark} An analogue of (\ref{eq:q_x(t)}) and (\ref{eq:p_x(t)}) holds in the event that $\omega =0$.
This is seen by proceeding as in the proof of Lemma~\ref{lem:harmsol} and observing that now
$\gamma(0) = 0$, but $\gamma(k) \neq 0$ for $k \neq 0$. For $k \neq 0$, the formulas above are correct, and
a simple calculation shows that, in this case,
\begin{equation}\label{eq:QPsolsw=0}
\begin{split}
Q_0(t) & =  Q_0(0) + 2P_0(0) t \\
P_0(t) &=  P_0(0),
\end{split}
\end{equation}
similar to (\ref{eq:QPsol}). One easily sees that the equations (\ref{eq:q_x(t)}) and (\ref{eq:p_x(t)}) still hold
with the convention that
\begin{equation}
h^{(-1)}_t(x) = - \frac{2t}{| \Lambda_L|} +  {\rm Im} \left[\frac{1}{|\Lambda_L|}\sum_{k \in \Lambda_L^* \setminus \{ 0 \}}\frac{e^{i(k \cdot x-2\gamma(k)t)}}{\gamma(k)} \right] \, .
\end{equation}
\end{remark}

We end this subsection with the following crucial estimate which was proven in \cite{harm}.
\begin{lemma}\label{lem:htx}
Fix $L\geq 1$ and consider the functions $h_t^{(m)}$ as defined in (\ref{eq:h}) for $m \in \{-1,0,1 \}$. 
For any $\mu >0$, the bounds
\begin{equation}
\begin{split}
\left| h_t^{(0)}(x) \right| &\leq  e^{-\mu \left( |x| - c_{\omega,\lambda} \max \left( \frac{2}{\mu} \, , \, e^{(\mu/2)+1}\right) |t| \right)}
\\
\left| h_t^{(-1)}(x) \right| &\le  c^{-1}_{\omega,\lambda}e^{-\mu \left( |x| - c_{\omega,\lambda} \max \left( \frac{2}{\mu} \, , \, e^{(\mu/2)+1}\right) |t| \right)}
\\
\left| h_t^{(1)}(x) \right| &\le c_{\omega,\lambda}e^{\mu/2}e^{-\mu \left( |x| - c_{\omega,\lambda} \max \left( \frac{2}{\mu} \, , \, e^{(\mu/2)+1}\right) |t| \right)}
\end{split}
\end{equation}
hold for all $t \in \mathbb{R}$ and $x \in \Lambda_L$. Here  $|x| = \sum_{j=1}^{\nu} |x_i|$ and one may take $c_{\omega,\lambda} = (\omega^2 + 4 \sum_{j=1}^{\nu} \lambda_j )^{1/2}$. 
\end{lemma}

We refer the interested reader to Lemma 3.7 of \cite{harm} for the proof. 
Moreover, we stress that Lemma~\ref{lem:htx} is valid for all $\omega \geq 0$.

%
%
%
%
%

\subsection{A general locality estimate} \label{subsec:genloc}

Our first locality bound for the harmonic Hamiltonian follows directly from Lemma~\ref{lem:harmsol} and Lemma~\ref{lem:htx}.
We state this as Theorem~\ref{thm:genharmbd} below. As we will see, Theorem~\ref{thm:harm} is an immediate consequence of Theorem~\ref{thm:genharmbd}.
Recall that we have defined $\mathcal{A}_{\Lambda_L}^{(1)}$ to be the set of observables $A \in \mathcal{A}_{\Lambda_L}$ for which:
given any $x \in \Lambda_L$, $\frac{\partial A}{\partial q_x} \in \mathcal{A}_{\Lambda_L}$, 
$\frac{\partial A}{\partial p_x} \in \mathcal{A}_{\Lambda_L}$, and
\begin{equation}
\| \partial A \|_{\infty} = \sup_{x \in \Lambda_L} \max \left( \left\| \frac{\partial A}{\partial q_x} \right\|_{\infty}, \left\| \frac{\partial A}{\partial p_x} \right\|_{\infty} \right) \, < \, \infty \, .
\end{equation}
\begin{thm} \label{thm:genharmbd}
Let $X$ and $Y$ be finite subsets of $\mathbb{Z}^{\nu}$ and take $L_0$ to be the smallest integer such that $X, Y \subset \Lambda_{L_0}$.
For any $L \geq L_0$, let $\alpha_t^{h,L}$ denote the dynamics corresponding to $H_h^{\Lambda_L}$. 
For any $\mu >0$ and any observables $A,B \in \mathcal{A}_{\Lambda_{L_0}}^{(1)}$ with 
support of $A$ in $X$ and support of $B$ in $Y$, the bound 
\begin{equation}\label{eq:thmbd2}
\left\| \left\{ \alpha_t^{h,L}(A), B \right\} \right\|_{\infty} \le C \| \partial A\|_{\infty} \| \partial B \|_{\infty} \sum_{x\in X, y \in Y} e^{-\mu\left(d(x,y) - c_{\omega, \lambda} \max (\frac{2}{\mu}, e^{(\mu/2) + 1})|t| \right)}
\end{equation}
holds for all $t \in \mathbb{R}$.  Here 
\begin{equation}\label{eq:d}
d(x,y) = \sum_{j=1}^{\nu} \min_{\eta_{j} \in \mathbb{Z}} |x_{j} - y_{j} + 2L \, \eta_{j}| 
\end{equation} 
is the distance on the torus and the constants may be taken as $C = (2 + c_{\omega, \lambda}e^{\mu/2} + c_{\omega, \lambda}^{-1})$ with $c_{\omega, \lambda} = (\omega^2 + 4\sum_{j = 1}^{\nu}\lambda_j)^{1/2}$.
\end{thm}

\begin{proof}
The Poisson bracket is easy to calculate. In fact, for any $\mathrm{x} \in \mathcal{X}_{\Lambda_L}$,
\begin{equation} \label{eq:genbd1}
\left[ \left\{ \alpha_t^{h,L} \left(A \right), B \right\} \right]( \mathrm{x})  =  \sum_{y \in Y} \frac{ \partial}{ \partial q_y} A \left( \Phi_t^{h,L}( \mathrm{x}) \right) \cdot \frac{ \partial B}{ \partial p_y}( \mathrm{x}) \, - \,  
\frac{ \partial}{ \partial p_y} A \left( \Phi_t^{h,L}( \mathrm{x}) \right) \cdot \frac{ \partial B}{ \partial q_y}( \mathrm{x}) \, .
\end{equation}
By the chain rule,
\begin{equation}
\frac{ \partial}{ \partial q_y} A \left( \Phi_t^{h,L}( \mathrm{x}) \right) = \sum_{x\in X} \frac{ \partial A}{ \partial q_x}\left( \Phi_t^{h,L}( \mathrm{x}) \right) \cdot \frac{ \partial q_x}{ \partial q_y}(t) \, + \,  
\frac{ \partial A}{ \partial p_x}\left( \Phi_t^{h,L}( \mathrm{x}) \right) \cdot \frac{ \partial p_x}{ \partial q_y}(t) \,
\end{equation}
and a similar formula holds for $\frac{ \partial}{ \partial p_y} A \left( \Phi_t^{h,L}( \mathrm{x}) \right) $. Now estimating (\ref{eq:genbd1}), we find that
\begin{equation}
\left\| \left\{ \alpha_t^{h,L} \left(A \right), B \right\} \right\|_{\infty} \, \leq \, \| \partial A \|_{\infty} \| \partial B \|_{\infty} \sum_{x \in X, y \in Y} \left| h_t^{(-1)}(x-y) \right| \, + \, 2 \, \left| h_t^{(0)}(x-y) \right|
\, + \, \left| h_t^{(1)}(x-y) \right|,
\end{equation}
using Lemma~\ref{lem:harmsol}. The bound in (\ref{eq:thmbd2}) now follows from Lemma~\ref{lem:htx}.
\end{proof}

{F}rom Theorem~\ref{thm:genharmbd}, and specifically the bound (\ref{eq:thmbd2}), we see that for any $\mu >0$, the
harmonic velocity $v_{\rm h}$ is essentially described by
\begin{equation} \label{eq:defvh}
v_{\rm h}( \mu) = c_{\omega, \lambda} \max \left( \frac{2}{\mu}, e^{(\mu/2)+1} \right) \, .
\end{equation} 
In fact, given (\ref{eq:thmbd2}) for some $\mu >0$, it is easy to see that for any $0 < \epsilon < 1$,
\begin{equation}
\sum_{x \in X, y \in Y} e^{- \mu d(x,y)} \leq e^{- \epsilon \mu d(X,Y)} \, \min \left(|X|, |Y| \right) \, \sum_{z \in \Lambda_L} e^{- \mu (1- \epsilon) d(0,z)} \, ,
\end{equation}
where we have set $d(X,Y) = \min_{x \in X, y \in Y}d(x,y)$. Thus, Theorem~\ref{thm:harm} is a simple consequence of Theorem~\ref{thm:genharmbd}. It is interesting to note that for any
$L$ the quantity
\begin{equation}
\sum_{z \in \Lambda_L} e^{- \mu (1- \epsilon) d(0,z)} \leq \sum_{z \in \mathbb{Z}^{\nu}} e^{- \mu (1- \epsilon) |z|},
\end{equation}
where $|z|$ denotes the $L^1$-metric on $\mathbb{Z}^{\nu}$. Given this and the fact that, for sufficiently large $L$, the distance
$d(X,Y)$ agrees with the $L^1$-distance between $X$ and $Y$, it is clear that the estimate proven in
Theorem~\ref{thm:harm} is genuinely independent of the length scale $L$.

Since the bounds are valid for any $\mu >0$, Theorem~\ref{thm:genharmbd} demonstrates arbitrarily fast 
exponential decay in space with a velocity that depends on $\mu$. Typically, however, one is interested in the best possible estimates
on $v_{\rm h}$ given some decay rate. In this sense, the optimal harmonic velocity, as described by
(\ref{eq:defvh}), occurs when the equation
\begin{equation} \label{eq:optv}
\frac{\mu}{2} = e^{(\mu/2)+1}
\end{equation}
holds. It is easy to see that the solution to (\ref{eq:optv}), denoted by $\mu_0$, satisfies $1/2 < \mu_0 <1$,
and therefore the corresponding velocity $v_{\rm h}( \mu_0) \leq 4 c_{\omega, \lambda}$.

%
%
%
%

\subsection{The harmonic evolution of Weyl functions} \label{subsec:weyl}
In preparation for our arguments in Sections \ref{sec:anharmss} and \ref{sec:anharmms}, we will
now present a different proof of our locality result, analogous to Theorem~\ref{thm:genharmbd}, 
valid for Weyl functions. Recall that a Weyl function is an observable, generated by a function
 $f : \Lambda_L \to \mathbb{C}$, with the form
 \begin{equation} \label{eq:defweyl}
 [W(f)]( \mathrm{x}) \, = \, \mbox{exp} \left[ \, i \, \sum_{x \in \Lambda_L} \mbox{Re}[f(x)] q_x + \mbox{Im}[f(x)] p_x \, \right] \, .
 \end{equation}
 
One important property of the Weyl functions is typically referred to as the Weyl relation.
We state this as Proposition~\ref{prop:wr}.
\begin{prop}[Weyl Relation] \label{prop:wr}
Let $f,g : \Lambda_L \to \mathbb{C}$. We have that
\begin{equation}\label{eq:ip}
\{ W(f), W(g) \} = -{\rm Im}[ \langle f, g \rangle ] W(f) W(g).
\end{equation}
where the inner product is taken in $\ell^2(\Lambda_L)$.
\end{prop}
\begin{proof}
A direct calculation yields
\begin{eqnarray*}
\{ W(f), W(g) \} &=&\sum_{x\in \Lambda_L} \frac{\partial}{\partial q_x}W(f) \,  \frac{\partial}{\partial p_x} W(g) -  \frac{\partial}{\partial p_x} W(f) \,  \frac{\partial}{\partial q_x} W(g)
\\
&=& \sum_{x\in \Lambda_L} (-\text{Re}[ f(x) ] \text{Im}[ g(x) ] + \text{Im}[ f(x) ]\text{Re}[ g(x) ])W(f)W(g).
\end{eqnarray*}
Noting that 
\begin{eqnarray*}
\text{Im}[\langle f, g \rangle ] &=& \text{Im}\big[\sum_{x\in \Lambda_L} \overline{f(x)} g(x) \big]
\\
&=& \sum_{x\in \Lambda_L}(-\text{Im}[ f(x) ] \text{Re}[ g(x) ] + \text{Re}[ f(x) ]\text{Im}[ g(x) ])
\end{eqnarray*}
proves the proposition.
\end{proof}

Another useful property of the Weyl functions is that the harmonic dynamics leaves this class of observables invariant. 
This important fact, which follows immediately from Lemma~\ref{lem:harmsol}, is the content of the next proposition.
Before stating this, it is convenient to introduce notation for the convolution of two functions $f,g: \Lambda_L \to \mathbb{C}$,
\begin{equation}\label{def:conv}
(f*g)(x) = \sum_{y \in \Lambda_L} f(y) g(x-y),
\end{equation}
where, if necessary, $g(x-y)$ is calculated by periodic extension.
\begin{prop}\label{prop:f_t}
Let $f: \Lambda_L \to \mathbb{C}$ and take $t \in \mathbb{R}$. 
\begin{equation}\label{eq:W(f_t)}
\alpha_t^{h,L}(W(f)) = W(f_t)  \, ,
\end{equation}
where 
\begin{equation} \label{eq:defft}
f_t = f * \overline{ \left(h_t^{(0)} + \frac{i}{2}(h_t^{(-1)} + h_t^{(1)}) \right)} + \overline{f}*\left( \frac{i}{2}(h_t^{(1)} - h_t^{(-1)}) \right).
\end{equation}
with $h_t^{(-1)}, h_t^{(0)},$ and $h_t^{(1)}$ as in (\ref{eq:h}).
\end{prop}

\begin{proof}
For any point $\mathrm{x} \in \mathcal{X}_{\Lambda_L}$, we have that
\begin{eqnarray}\label{eq:W(f,t)}
\left[ \alpha_t^{h,L}(W(f))\right](\mathrm{x}) &=& \exp \left( i\sum_{x \in \Lambda_L} \text{Re}[ f(x) ]q_x(t) + \text{Im}[ f(x) ]p_x(t) \right) \nonumber
\\
&=& \exp \left( i\sum_{x \in \Lambda_L} \text{Re}[ f(x) ] \sum_{y \in \Lambda_L} q_y(0) \, h^{(0)}_t(x-y) - p_y(0) \, h^{(-1)}_t(x-y) \right) 
\nonumber \\
& \mbox{ } & \quad \times \exp \left(i \sum_{x \in \Lambda_L} \text{Im}[ f(x) ]\sum_{y \in \Lambda_L} q_y(0) \, h^{(1)}_t(x-y) + p_y(0) \, h^{(0)}_t(x-y) \right)
 \nonumber
\\
&=& \exp \left( i \sum_{y \in \Lambda_L} q_y(0) \sum_{x \in \Lambda_L} \text{Re}[ f(x) ]  h^{(0)}_t(x-y) + \text{Im}[ f(x) ] \, h^{(1)}_t(x-y) \right) 
\nonumber \\
& \mbox{ } & \quad \times \exp \left(i \sum_{y \in \Lambda_L} p_y(0) \sum_{x \in \Lambda_L} \text{Im}[ f(x) ] \, h^{(0)}_t(x-y)  - \text{Re}[ f(x) ] h^{(-1)}_t(x-y) \right)
\nonumber \\
& = & \left[ W(f_t) \right]( \mathrm{x}),
\end{eqnarray}
where we have defined the function $f_t : \Lambda_L \to \mathbb{C}$ by (\ref{eq:defft}).
\end{proof}

It is obvious that Theorem~\ref{thm:weylharmbd} below follows immediately from 
Theorem~\ref{thm:genharmbd}, since the Weyl functions are clearly in $\mathcal{A}_{\Lambda_L}^{(1)}$.
We will here give a different, but equally short, proof which uses the specific properties of Weyl functions.
\begin{thm}\label{thm:weylharmbd}
Let $X$ and $Y$ be finite subsets of $\mathbb{Z}^{\nu}$ and take $L_0$ to be the smallest integer such that $X, Y \subset \Lambda_{L_0}$.
For any $\mu >0$, $L \geq L_0$, and any functions $f,g: \Lambda_{L_0} \to \mathbb{C}$ with 
support of $f$ in $X$ and support of $g$ in $Y$, the bound 
\begin{equation}\label{eq:thmbd3}
\left\| \left\{ \alpha_t^{h,L}(W(f)), W(g) \right\} \right\|_{\infty} \le C \|f\|_{\infty} \|g\|_{\infty} \sum_{x\in X, y \in Y} e^{-\mu\left(d(x,y) - c_{\omega, \lambda} \max (\frac{2}{\mu}, e^{(\mu/2) + 1})|t| \right)}
\end{equation}
holds for all $t \in \mathbb{R}$.  Here, as in (\ref{eq:d}), $d(x,y)$ is the distance on the torus and the constants 
may be taken as $C = (1 + c_{\omega, \lambda}e^{\mu/2} + c_{\omega, \lambda}^{-1})$ with $c_{\omega, \lambda} = (\omega^2 + 4\sum_{j = 1}^{\nu}\lambda_j)^{1/2}$.
\end{thm}

\begin{proof}  
Combining Propositions \ref{prop:f_t} and \ref{prop:wr}, it is clear that
\begin{equation}
\left\{ \alpha_t^{h,L}(W(f)), W(g) \right\} = \left\{ W(f_t), W(g) \right\} = - {\rm Im}\left[ \langle f_t, g \rangle \right] W(f_t) W(g).
\end{equation}
In this case, the bound
\begin{equation}
\left\| \left\{ \alpha_t^{h,L}(W(f)), W(g) \right\} \right\|_{\infty} \leq \left|  {\rm Im} \left[ \langle f_t, g \rangle \right]  \right|,
\end{equation}
readily follows. Appealing again to Proposition \ref{prop:f_t}, we have that for any $y \in \Lambda_L$, 
\begin{eqnarray} \label{eq:fty}
f_t(y) & = & \sum_{x \in X} f(x) \left( h^{(0)}_t(x-y) - \frac{i}{2}h^{(-1)}_t(x-y) - \frac{i}{2}h^{(1)}_t(x-y) \right) \nonumber \\
& \mbox{ } & \quad + \sum_{x \in X} \overline{f(x)} \left( \frac{i}{2}h^{(1)}_t(x-y) - \frac{i}{2}h^{(-1)}_t(x-y) \right) \, ,
\end{eqnarray}
and therefore,
\begin{eqnarray}
\left| \langle f_t, g \rangle \right|  & = &  \left| \sum_{y \in Y} \overline{f_t(y)}g(y) \right|  \nonumber \\ 
& \leq &  \| f \|_{\infty} \, \| g \|_{\infty} \sum_{x \in X, y \in Y} |h^{(0)}_t(x-y)| + |h^{(-1)}_t(x-y)| + |h^{(1)}_t(x-y)| \,.
\end{eqnarray}
Theorem \ref{thm:weylharmbd} now follows from Lemma \ref{lem:htx}. 
\end{proof}

We end this section with a corollary of Theorem~\ref{thm:weylharmbd} that will be particularly
useful in the next sections. The locality bound we prove for the anharmonic dynamics is derived by iterating 
a certain inequality involving the harmonic estimate. With this in mind, it is useful to introduce the following family of
decaying functions. For any $\mu >0$, consider $F_{\mu} : [0, \infty) \to (0, \infty)$ 
defined by 
\begin{equation} \label{eq:defF}
F_\mu (r) = \frac{e^{-\mu r}}{(1+r)^{\nu+1}} \, .
\end{equation}
Clearly, these function $F_{\mu}$ also depend on the quantity $\nu \geq 1$, which is the dimension
of the underlying lattice in our models, but we will suppress that dependence in our notation. 
Unlike the bare exponential $e^{- \mu r}$, these functions have the following nice property.
There exists a number $C_{\nu} >0$ for which, given any pair of sites $x,y \in \mathbb{Z}^{\nu}$,
\begin{equation}\label{eq:Fconv}
\sum_{z\in \mathbb{Z}^{\nu}} F_\mu( |x-z| )F_\mu( |z-y| ) \leq C_\nu F_\mu( |x -y| ) \, .
\end{equation}
Here one may take 
\begin{equation}\label{eq:Cnu}
C_\nu = 2^{\nu+1}\sum_{z \in \mathbb{Z}^{\nu} }\frac{1}{(1+|z|)^{\nu+1}} \, .
\end{equation}
Functions of this type were introduced in \cite{NaOgSi}, see also \cite{harm}, as an aide in proving Lieb-Robinson bounds. 
We will use them here as well.  

We can rewrite the decay expressed in our harmonic estimates, i.e. (\ref{eq:thmbd2}), in terms of these functions $F_{\mu}$.  
\begin{corollary}\label{cor:weylharmbd}
Let $X$ and $Y$ be finite subsets of $\mathbb{Z}^{\nu}$ and take $L_0$ to be the smallest integer such that $X, Y \subset \Lambda_{L_0}$.
For any $\mu >0$, $\epsilon>0$, $L \geq L_0$, and any functions $f,g: \Lambda_{L_0} \to \mathbb{C}$ with 
support of $f$ in $X$ and support of $g$ in $Y$, the bound 
\begin{equation}\label{eq:harmbdF}
\left\| \left\{ \alpha_t^{h,L}(W(f)), W(g) \right\} \right\|_{\infty}  \leq C \, \|f\|_{\infty}\,  \|g\|_{\infty} \, e^{(\mu+\epsilon) v_{\mathrm{h}} (\mu+\epsilon)\, |t|} \sum_{x \in X, y\in Y} F_{\mu}(d(x,y)) \, , 
\end{equation}
holds for all $t \in \mathbb{R}$.  Here
\begin{equation}
C \, = \,  (1 + c_{\omega, \lambda} e^{\frac{( \mu + \epsilon)}{2}} + c_{\omega, \lambda}^{-1} ) \, \sup_{s \geq 0}  \left[ (1+s)^{\nu + 1} e^{- \epsilon s} \right] \, .
\end{equation}
and $v_{\rm{h}}$ is as defined in (\ref{eq:defvh}).
\end{corollary}

%
%
%
%

\setcounter{equation}{0}

\section{Single Site Anharmonicities} \label{sec:anharmss}

In this section, we will prove a locality result, analogous to Theorem~\ref{thm:weylharmbd}, for a
specific class of perturbations of the harmonic Hamiltonian. A much more general result, which
follows from the same basic arguments, is presented in the next section. We begin with
a precise statement of the models we consider, and then prove the result.

To make our basic technique more transparent, we will only consider single-site potentials that are generated by a particular function
$V$ in this section, see Section~\ref{sec:anharmms} for a more general result. Let $V : \mathbb{R} \to \mathbb{R}$ satisfy
$V \in C^2(\mathbb{R})$,  $V' \in L^1(\mathbb{R})$,  $V'' \in L^{\infty}(\mathbb{R})$, and suppose further that the
quantity
\begin{equation} \label{eq:defkap}
\kappa_V = \int_{\mathbb{R}} |r| \, |\widehat{V'} (r)| \,  \rd r 
\end{equation}
is finite. Here $\widehat{V'}$ is the Fourier transform of $V'$. Given such a function $V$ and an integer $L \geq 1$, we
define an anharmonic Hamiltonian $H^{\Lambda_L} : \mathcal{X}_{\Lambda_L} \to \mathbb{R}$ by setting
 \begin{equation}\label{eq:Anhar}
 H^{\Lambda_L} = H_h^{\Lambda_L} + \sum_{z\in \Lambda_L} V_z \, ,
\end{equation}
where for each $z \in \Lambda_L$, the function $V_z : \mathcal{X}_{\Lambda_L} \to \mathbb{R}$ is given by $V_z(\mathrm{x}) = V(q_z)$.

As is discussed at the end of Section~\ref{subsec:weyl}, we will state our result in terms of the functions $F_{\mu} : [0, \infty) \to (0, \infty)$
given by
\begin{equation} \label{eq:defF2}
F_\mu (r) = \frac{e^{-\mu r}}{(1+r)^{\nu+1}} \, ,
\end{equation}
with $\nu>0$ corresponding to the dimension of the underlying lattice $\mathbb{Z}^{\nu}$.
The goal of this section is to prove the following result. 
\begin{thm}\label{thm:anharm2}
Suppose $V : \mathbb{R} \to \mathbb{R}$ satisfies
$V \in C^2(\mathbb{R})$,  $V' \in L^1(\mathbb{R})$,  $V'' \in L^{\infty}(\mathbb{R})$, and $\kappa_V$, as in (\ref{eq:defkap}) above, is
finite. Let $X$ and $Y$ be finite subsets of $\mathbb{Z}^{\nu}$ and take $L_0$ to be the smallest integer such that $X, Y \subset \Lambda_{L_0}$.
For any $L \geq L_0$ and $t \in \mathbb{R}$, let $\alpha_t^L$ denote the dynamics corresponding to $H^{\Lambda_L}$.
Then, for any $\mu >0$, $\epsilon >0$, and any functions $f,g: \Lambda_{L_0} \to \mathbb{C}$ with 
support of $f$ in $X$ and support of $g$ in $Y$, the bound 
\begin{equation}\label{eq:anharmbd}
 \left\| \left\{ \, \alpha_t^L (W(f)) , W (g) \, \right\} \right\|_{\infty} \leq  C \, \|f\|_{\infty} \| g \|_{\infty} \, e^{ \delta |t|} 
\sum_{x\in X, y\in Y} \, F_{\mu} \left(d(x,y) \right)
\end{equation}
holds for all $t \in \mathbb{R}$.
Here one may take 
\begin{equation} \label{eq:constc}
C \, = \,  (1 + c_{\omega, \lambda} e^{\frac{( \mu + \epsilon)}{2}} + c_{\omega, \lambda}^{-1} ) \, \sup_{s \geq 0}  \left[ (1+s)^{\nu + 1} e^{- \epsilon s} \right] \, 
\end{equation}
and 
\begin{equation} \label{eq:vah}
\delta = \delta(\mu, \epsilon) =  ( \mu + \epsilon )v_{\rm h}( \mu + \epsilon) + C C_{\nu} \kappa_V
\end{equation}
where $v_{\rm h}$ is as in (\ref{eq:defvh}), $C_{\nu}$ is in (\ref{eq:Cnu}), and $\kappa_V$ is in (\ref{eq:defkap}).
\end{thm}

Before we prove Theorem~\ref{thm:anharm2}, we comment on the corresponding
anharmonic velocity. With arguments similar to those given after the proof of Theorem~\ref{thm:genharmbd}, it is
clear that Theorem~\ref{thm:anharm2} implies Theorem~\ref{thm:anharm}. 
In this case, we find that an upper bound on the anharmonic velocity for this model is
\begin{equation}
v_{\rm ah}( \mu, \epsilon)  = \left( 1 + \frac{ \epsilon}{ \mu} \right) v_{\rm h}( \mu + \epsilon) + \frac{C C_{\nu} \kappa_V}{ \mu} \, .
\end{equation}

We now present the proof.

\noindent{\it{Proof of Theorem~\ref{thm:anharm2}. }}  
Our proof of this estimate is perturbative, and
we begin by interpolating between the harmonic and anharmonic dynamics.
Fix $L \geq L_0$ as in the statement of the theorem. Since we will regard both the
harmonic and anharmonic dynamics on the same volume $\Lambda_L$, we drop 
the dependence of each on $L$. Observe that for any $t>0$, 
\begin{equation}
\left\{ \alpha_t(W(f)), W(g) \right\} \, - \, \left\{ \alpha^h_t(W(f)), W(g) \right\}  \, = \, \int_0^t \frac{\rd}{\rd s}  \left\{ \alpha_s \left( \alpha_{t-s}^h (W(f)) \right), W(g) \right\} \, \rd s \, .
\end{equation}
Moreover, a direct calculation shows that
\begin{eqnarray} \label{eq:dalpha}
\frac{\rd}{\rd s} \alpha_s \left( \alpha_{t-s}^h \left( W(f) \right) \right) & = & \alpha_s \left( \left\{ \alpha_{t-s}^h(W(f)), H \right\} \right) - \alpha_s \left( \alpha_{t-s}^h \left( \left\{ W(f), H_h \right\} \right) \right) \nonumber \\
& = & \alpha_s \left( \left\{ \alpha_{t-s}^h(W(f)), H-H_h \right\} \right) \nonumber \\
& = & \sum_{z \in \Lambda_L} \alpha_s \left( \left\{ \alpha_{t-s}^h(W(f)), V_z \right\} \right) \, .
\end{eqnarray}
The Poisson bracket on the right-hand side of (\ref{eq:dalpha}) 
can be simplified
\begin{equation}
\left\{ \alpha_{t-s}^h(W(f)), V_z \right\}  = \left\{W(f_{t-s}), V_z \right\} = - i {\rm Im} \left[f_{t-s}(z) \right] \, W(f_{t-s}) \, V_z' \,  .
\end{equation}
For the first equality above we used Proposition~\ref{prop:f_t}, and we have denoted by
$V_z'$ the function $V_z': \mathcal{X}_{\Lambda_L} \to \mathbb{R}$ with $V_z'( \mathrm{x}) = V'(q_z)$. 

These calculations lead to a particularly simple differential equation and thus, eventually, the bound (\ref{eq:bd1})
appearing below. In fact, for $t>0$ fixed
and $0 \leq s \leq t$, define the function
\begin{equation}\label{eq:psi}
\Psi_t(s) = \{ \alpha_s ( \alpha_{t-s}^h (W(f))), W(g) \} \, .
\end{equation}
We have shown that
\begin{eqnarray} \label{eq:dpsi}
\frac{\rd}{\rd s} \Psi_t(s) & = &  \sum_{z \in \Lambda_L} \left\{ \alpha_s \left( \left\{ \alpha_{t-s}^h(W(f)), V_z \right\} \right), W(g) \right\} \nonumber \\
& = & i \mathcal{L}_t(s) \Psi_t(s) \, + \, \mathcal{Q}_t(s),  
\end{eqnarray}
where
\begin{equation}
\begin{split}
\mathcal{L}_t(s) & =   - \sum_{z \in \Lambda_L} {\rm Im} \left[ f_{t-s}(z) \right] \alpha_s( V_z' ) \, ,  \\
\mathcal{Q}_t(s) & =  -  i \sum_{z \in \Lambda_L}  {\rm Im} \left[ f_{t-s}(z) \right] \alpha_s \left( \alpha_{t-s}^h(W(f)) \right) \left\{ \alpha_s ( V_z' ) , W(g) \right\}\, ,
\end{split}
\end{equation}
and the final equality in (\ref{eq:dpsi}) follows from an application of the Leibnitz rule for Poisson brackets.
Since for each fixed $s$, $\mathcal{L}_t(s)$ is a real-valued function of phase space, the solution $U_t$ of
\begin{equation}
\frac{\rd}{\rd s} U_t(s) \, = \, -i \mathcal{L}_t(s) U_t(s) \quad \mbox{with } U_t(0) = 1,
\end{equation}
is a complex exponential. In addition, it is easy to see that
\begin{equation}
\frac{\rd}{\rd s} \left( \Psi_t(s) U_t(s) \right) \, = \, \mathcal{Q}_t(s) U_t(s),
\end{equation}
and therefore,
\begin{equation}
\Psi_t(t) U_t(t) \, = \, \Psi_t(0) \, + \, \int_0^t \mathcal{Q}_t(s) U_t(s) \, \rd s,
\end{equation}
from which 
\begin{eqnarray} \label{eq:bd1}
\left\| \left\{ \alpha_t(W(f)), W(g) \right\} \right\|_{\infty} & \leq & \left\| \left\{ \alpha_t^h(W(f)), W(g) \right\} \right\|_{\infty}  \nonumber \\
& \mbox{ } & + \sum_{z \in \Lambda_L} \int_0^t \, \left| {\rm Im} \left[ f_{t-s}(z) \right] \right|  \left\| \left\{ \alpha_s ( V_z' ) , W(g) \right\} \right\|_{\infty} \, \rd s \, ,
\end{eqnarray}
readily follows. 

Now, if $V_z'$ was a Weyl function, then we could immediately iterate the inequality in (\ref{eq:bd1}) and
derive a bound. This is not the case, however, our assumptions on $V$ allow us to write $V_z'$ as
an average of Weyl functions through its Fourier representation. In fact,  we write the Fourier transform of
$V'$ as
\begin{equation}
\widehat{V' }(r) = \frac{1}{2\pi} \int_{\mathbb{R}}  e^{-iqr} \,  V'(q) \, \rd q \, ,
\end{equation}
and by inversion, one has that
\begin{equation} \label{eq:finver}
V'(q) =  \int_{\mathbb{R}} e^{irq} \, \widehat{ V'}(r) \,  \rd r \,. 
\end{equation}
This implies that, as a function of phase space, $V_z'$ can be expressed as
\begin{equation}\label{eq:V_z'}
V_z' = \int_{\mathbb{R}} W(r\delta_z) \widehat{V'}(r) \, \rd r \, 
\end{equation} 
where $r \delta_z : \Lambda_L \to \mathbb{R}$ is the function that has value $r$ at $z$ and $0$ otherwise. 
Inserting (\ref{eq:V_z'}) into (\ref{eq:bd1}), we have that
\begin{eqnarray} \label{eq:bd2}
\left\| \left\{ \alpha_t(W(f)), W(g) \right\} \right\|_{\infty} & \leq &  \left\| \left\{ \alpha_t^h(W(f)), W(g) \right\} \right\|_{\infty}  \\
& +&  \sum_{z \in \Lambda_L} \int_0^t  \, \left| {\rm Im} \left[ f_{t-s}(z) \right] \right|  \int_{\mathbb{R}} \left| \widehat{ V'}(r) 
\right| \big\| \left\{ \alpha_{s} \left( W(r \delta_z) \right) , W(g) \right\} \big\|_{\infty} \, \rd r \, \rd s \,. \nonumber
\end{eqnarray}
At this stage, we can finally iterate the inequality. First, however, we insert the harmonic bound found in
Corollary~\ref{cor:weylharmbd}.

Recall that for any $\mu >0$ and $\epsilon >0$ we have established (\ref{eq:harmbdF}) with a constant $C$ as in (\ref{eq:constc}). 
With equation (\ref{eq:fty}), it is easy to see that, for any $\mu >0$ and $\epsilon >0$
\begin{equation}\label{eq:FH}
\left| {\rm Im}[f_t(z)] \right| \leq C \, \|f\|_{\infty} \, e^{(\mu + \epsilon) v_{\mathrm{h}}( \mu + \epsilon)|t|} \sum_{x \in X} F_{\mu}(d(x, z)) \, ,
\end{equation}
also holds for any $z \in \Lambda_L$ and $t \in \mathbb{R}$.
To ease the notation a bit, we will denote by $\tilde{v} = (\mu + \epsilon) v_{\mathrm{h}}( \mu + \epsilon)$.
Using these bounds, the inequality in (\ref{eq:bd2}) now takes the form
\begin{eqnarray}\label{eq:bdtoit}
\begin{split}
&\left\| \left\{ \alpha_t(W(f)), W(g) \right\} \right\|_{\infty}   \leq  C \, \| f \|_{\infty} \, \| g \|_{\infty}
\, e^{\tilde{v} t} \, \sum_{x \in X, y \in Y} F_{\mu}(d(x,y))  
\\
& + C \, \|f\|_{\infty} \sum_{x\in X} \sum_{ z \in \Lambda_L}  F_{\mu}(d(x,z)) \int_0^t  \, e^{\tilde{v}(t-s)} \, \int_{\mathbb{R}} | \widehat{ V'}(r) 
| \big\| \left\{ \alpha_s \left( W(r \delta_z) \right) , W(g) \right\} \big\|_{\infty} \, \rd r \, \rd s  \, .
\end{split}
\end{eqnarray}

Upon iterating (\ref{eq:bdtoit}) $m \geq 1$ times, we find that
\begin{equation}
\left\| \left\{ \alpha_t(W(f)), W(g) \right\} \right\|_{\infty}  \leq C \, \| f \|_{\infty} \, \| g \|_{\infty}
\, e^{\tilde{v} t} \, \sum_{x \in X, y \in Y} \sum_{n=0}^m a_n(x,y;t) \, + \, R_{m+1}(t), 
\end{equation}
where 
\begin{equation} \label{eq:a0}
a_0(x,y;t) = F_{\mu}(d(x,y)) \, ,
\end{equation}
\begin{eqnarray}  \label{eq:a1}
a_1(x,y;t)  & = &   C t \, \int_{\mathbb{R}} |r| \, \left| \widehat{V'}(r) \right| \, dr \,  \sum_{z \in \Lambda_L} F_{\mu} \left( d(x,z) \right) F_{\mu} \left( d(z, y) \right) \nonumber \\
& \leq & C \, \kappa_V \, C_{\nu} \, t F_{\mu} \left( d(x,y) \right) \, ,
\end{eqnarray}
and in general,
\begin{eqnarray} \label{eq:an}
a_n(x,y;t) & = & \frac{(C t)^n}{n!} \left( \prod_{k=1}^n \int_{\mathbb{R}} |r_k| \, \left| \widehat{V'}(r_k) \right| \, dr_k \right) \sum_{z_1, \cdots, z_n \in \Lambda_L} F_{\mu} \left( d(x,z_1) \right) \cdots F_{\mu} \left( d(z_n, y) \right) \nonumber \\
& \leq & \frac{(C\, \kappa_V \, C_{\nu} \, t)^n}{n!} F_{\mu} \left( d(x,y) \right) \, ,
\end{eqnarray}
for any $1 \leq n \leq m$. In (\ref{eq:a1}) and (\ref{eq:an}), we have used (\ref{eq:Fconv}) several times.

{F}rom Lemma~\ref{lem:pbbd}, found in Section~\ref{sec:app}, it is easy to see that the apriori estimate 
\begin{equation}
\left\|  \left\{ \alpha_s(W(r \delta_z)), W(g) \right\} \right\|_{\infty} \, \leq \, C_1 \, | Y | \, |r| \, \| g \|_{\infty} \, {\rm exp} \left( C_2 \, t^2 \right)
\end{equation}
holds for all $0 \leq s \leq t$. Thus, for $t>0$ fixed, the remainder term $R_{m+1}(t)$ converges to zero as $m \to \infty$. 
In fact,
\begin{equation}\label{eq:error}
R_{m+1}(t)  \leq C_1 \, |X| \, |Y| \, \| f \|_{\infty} \| g \|_{\infty} e^{ \tilde{v} t + C_2 t^2} \frac{(C \, \kappa_V \, C_{\nu} t )^{m+1}}{(m+1)!} \, .  
 \end{equation}
We have proven that
\begin{equation}\label{eq:LR2}
\left\| \left\{ \alpha_t(W(f)), W(g) \right\} \right\|_{\infty}   \leq  C \, \| f \|_{\infty} \, \| g \|_{\infty}
\, e^{(\tilde{v} + C \, \kappa_V \, C_{\nu} ) t} \, \sum_{x \in X, y \in Y} F_{\mu}(d(x,y)) \, , 
\end{equation}
i.e. (\ref{eq:anharmbd}) as claimed. $\hfill \qed$

%
%
%
%

\setcounter{equation}{0}

\section{Multiple Site Anharmonicities} \label{sec:anharmms}

In this section, we will generalize Theorem~\ref{thm:anharm2} in such a way that it covers
perturbations involving long range interactions. As in the previous sections, we will be fixing
some integer $L \geq 1$ and considering only finite volumes $\Lambda_L \subset \mathbb{Z}^{\nu}$.

We will introduce these perturbations quite generally and then discuss the assumptions necessary to
prove our locality result. To each subset $Z \subset \Lambda_L$, we will assign a function 
$V( \cdot; Z) : \mathbb{R}^Z \to \mathbb{R}$ and a corresponding function of phase space
$V_Z : \mathcal{X}_{\Lambda_L} \to \mathbb{R}$ defined by setting
\begin{equation}
V_Z( \mathrm{x}) = V \left( \{ q_z \}_{z \in Z}; Z \right) \, .
\end{equation}
Here $\{ q_z \}_{z \in Z}$ is regarded as a vector in $\mathbb{R}^{Z}$ and the number $V \left( \{ q_z \}_{z \in Z}; Z \right)$
is calculated by evaluating $V( \cdot; Z)$ with $q_z$ as the value in the $z$-th component for each $z \in Z$.
With this understanding, we will use the notation
\begin{equation}
\partial_z V_Z = \frac{ \partial}{ \partial q_z} V_Z = \partial_zV( \cdot ; Z) \, ,
\end{equation}
to denote the partial derivatives of $V_Z$.

In general, the finite volume anharmonic Hamiltonians we consider are of the form $H^{\Lambda_L} : \mathcal{X}_{\Lambda_L} \to \mathbb{R}$
with
\begin{equation} \label{eq:defgenham}
H^{\Lambda_L} = H^{\Lambda_L}_h + \sum_{Z \subset \Lambda_L} V_Z \, ,
\end{equation}
where the sum above is over all subsets of $\Lambda_L$. As we saw in Section~\ref{sec:anharmss}, in order to prove our locality 
result, we need some assumptions on the functions $V_Z$. We will now list these explicitly below.

First, we use Lemma~\ref{lem:solbd}, proven in Section~\ref{sec:app}, to provide explicit bounds on the components of the flow which, in particular,
prevent the solutions from blowing-up in finite time. For these estimates, we assume the perturbation above satisfies: 
\newline i) For each $Z \subset \Lambda_L$, the function $V_Z$ has well-defined first
order partial derivatives. 
\newline ii) There exist numbers $C_1 \geq 0$, $\tilde{C}_1 \geq 0$, and $\mu_1 \geq 0$ such that
for each $x \in \Lambda_L$ and any $\mathrm{x} \in \mathcal{X}_{\Lambda_L}$,
\begin{equation} \label{eq:harmdom0}
\left( \sum_{Z \subset \Lambda_L}  \left| \partial_x V_Z ( \mathrm{x}) \right| \right)^2 \leq  C_1 \sum_{y \in \Lambda_L} (q_y^2 + \tilde{C}_1) F_{\mu_1} \left(d(x,y) \right) \, .
\end{equation}
The decaying functions $F_{\mu}$ are as defined at the end of Section~\ref{subsec:weyl}. 

Next, much like in the proof of Theorem~\ref{thm:anharm2}, we will need an apriori estimate on the
Poisson bracket of specific, dynamically evolved observables. This is the content of Lemma~\ref{lem:pbbd}
found in the next section. To prove it we use Lemma~\ref{lem:dsolbd}, and therefore, we must assume
\newline iii) For each $Z \subset \Lambda_L$, the function $V_Z$ has well-defined second order 
partial derivatives.
\newline iv) There exist numbers $C_2 \geq 0$ and $\mu_2 \geq 0$ for which: given any
pair $x,y \in \Lambda_L$, the bound
\begin{equation} \label{eq:2derbd0}
\sum_{Z \subset \Lambda_L} \left| \left[ \partial_x \partial_y V_Z \right] ( \mathrm{x}) \right| \, \leq \, C_2 F_{\mu_2} \left( d(x,y) \right) \, ,
\end{equation}
holds for all points $\mathrm{x} \in \mathcal{X}_{\Lambda_L}$.

Lastly, we need the quantities that arise in our iteration scheme to be well-defined.
For this we assume
\newline v) For each $Z \subset \Lambda_L$, the first order partial derivatives of $V_Z$ are integrable.
By this we mean that given $Z \subset \Lambda_L$ and $z \in Z$, the function $\partial_zV( \cdot ; Z)$ is
in $L^1( \mathbb{R}^Z)$ with respect to Lebesgue measure. In this case, the Fourier transform of these
functions exists, and we will write
\begin{equation}
\widehat{\partial_zV}(r;Z) = \frac{1}{(2 \pi)^{|Z|}} \int_{\mathbb{R}^Z} e^{-i r \cdot q} \, \partial_zV(q;Z)\, \rd q \, ,
\end{equation}
for any $r \in \mathbb{R}^Z$.
\newline vi) For each $Z \subset \Lambda_L$, we assume that the Fourier inversion formula holds for all
first order partial derivatives of $V_Z$. Thus, for any $q \in \mathbb{R}^Z$,
\begin{equation}
\partial_zV(q ; Z) = \int_{\mathbb{R}^Z} e^{i r \cdot q} \, \widehat{\partial_zV}(r ; Z) \, \rd r \, ,
\end{equation}
and therefore, we will write
\begin{equation} \label{eq:inversion}
\left[ \partial_z V_Z \right] ( \mathrm{x}) = \int_{\mathbb{R}^Z} \left[W(r \cdot \delta_Z) \right]( \mathrm{x}) \, \widehat{ \partial_zV}(r ; Z) \, \rd r \, ,
\end{equation}
where the function $r \cdot \delta_Z : \Lambda_L \to \mathbb{R}$ is given by
\begin{equation}
[r \cdot \delta_Z](x) = \sum_{z \in Z} r_z \, \delta_z(x) \quad \mbox{hence} \quad  \left[W(r \cdot \delta_Z) \right]( \mathrm{x}) = \exp \left[ i \sum_{z \in Z} r_z q_z \right] \, ,
\end{equation}
as required.
\newline vii) There exist numbers $C_3 \geq 0$ and $\mu_3 \geq 0$ such that given any points
$x,y \in \Lambda_L$, the bound
\begin{equation} \label{eq:fass}
\sum_{\substack{Z \subset \Lambda_L \\x,y \in Z}} \int_{\mathbb{R}^Z} |r| \cdot \left| \widehat{ \nabla V}(r ; Z) \right| \, dr \, \leq \, C_3 \, F_{\mu_3} \left( d(x,y) \right) \, .
\end{equation}
Here the vector-valued function $\widehat{\nabla V}( \cdot ; Z) : \mathbb{R}^Z \to \mathbb{C}^Z$ has components 
$\widehat{\partial_zV}( \cdot; Z)$ for each $z \in Z$. The number $|r|$, corresponding to some $r \in \mathbb{R}^Z$, is taken
as $|r| = \sum_{z \in Z}|r_z|$, but, as is seen in the proof below, any norm on $\mathbb{R}^Z$ satisfying $|r_z| \leq\| r \|$ will suffice.
As will become apparent below, we interpret the function $F_{\mu_3}$ in assumption vii) as our crucial 
estimate on the range of the interactions.

We now state our most general result.
\begin{thm}\label{thm:anharmms}
Let $X$ and $Y$ be finite subsets of $\mathbb{Z}^{\nu}$ and take $L_0 \geq 1$ to be the smallest integer
such that $X,Y \subset \Lambda_{L_0}$. For any $L \geq L_0$ and $t \in \mathbb{R}$, let $\alpha_t^L$ denote the
dynamics corresponding to the anharmonic Hamiltonian $H^{\Lambda_L}$ in (\ref{eq:defgenham}), and suppose that the 
perturbation satisfies assumptions i) - vii) listed above.  Then, for each $\epsilon >0$ and any functions 
$f, g \colon \Lambda_{L_0} \rightarrow \mathbb{C}$ with the support of  $f$ in $X$ and the support of  $g$ in $Y$, 
\begin{equation}\label{eq:anharmbdms}
 \left\| \left\{ \, \alpha_t^L (W(f)) , W (g) \, \right\} \right\|_{\infty} \leq  C \, \|f\|_{\infty} \| g \|_{\infty} \,e^{ \delta  |t|} \sum_{x\in X, y\in Y} \, F_{\mu_3}(d(x,y))
\end{equation}
holds for any $t \in \mathbb{R}$.   Here
\begin{equation} \label{eq:fc}
 C \, = \,  (1 + c_{\omega, \lambda} e^{\frac{( \mu_3 + \epsilon)}{2}} + c_{\omega, \lambda}^{-1} ) \, \sup_{s \geq 0}  \left[ (1+s)^{\nu + 1} e^{- \epsilon s} \right] \, .
\end{equation}
and
\begin{equation}\label{eq:vah2}
\delta = \delta( \epsilon) = (\mu_3+\epsilon) \, v_{\rm h}(\mu_3 + \epsilon) + C \, C_3 \, C_{\nu}^2  \, ,
\end{equation}
where $v_{\rm h}$ is as in (\ref{eq:defvh}), $C_{\nu}$ is in (\ref{eq:Cnu}), and $C_3$ is in (\ref{eq:fass}).
\end{thm}

One important difference between the bound we prove in Theorem~\ref{thm:anharmms} above, in contrast to the
one proven in Theorem~\ref{thm:anharm2}, is that the spatial decay rate in (\ref{eq:anharmbdms}) can be no greater
than the rate $\mu_3$ appearing in (\ref{eq:fass}). If $\mu_3 >0$, then there is a corresponding velocity for this
anharmonic system
\begin{equation}
v_{\rm ah}( \epsilon) = \left( 1 + \frac{ \epsilon}{ \mu_3} \right) v_{\rm h}( \mu_3 + \epsilon) + \frac{C \, C_3 \, C_{\nu}^2}{ \mu_3} \, .
\end{equation}
Since the case of $\mu_3 = 0$ represents only polynomial decay in the interaction range, as measured by
(\ref{eq:fass}), the bound in (\ref{eq:anharmbdms}) at most decays polynomially in distance between the supports of $f$ and $g$ as well.

\begin{example}
To clarify the general assumptions on the perturbation introduced above, we will
consider a simple model with pair interactions generated by a single function. One can compare 
this example with the single site, anharmonic Hamiltonian analyzed in Section~\ref{sec:anharmss}. 
Let $V : \mathbb{R}^2 \to \mathbb{R}$ be given and fix some number $\mu \geq 0$. 
For each $L \geq 1$ and any $Z \subset \Lambda_L$, define
\begin{equation}
V( \cdot ; Z) \, = \, \left\{ \begin{array}{cc} F_{\mu} \left( d(z_1, z_2 ) \right) V( \cdot ) & \mbox{if } Z = \{ z_1, z_2 \} \, , \\
0 & \mbox{otherwise},
\end{array} \right.
\end{equation}
and thereby, the anharmonic Hamiltonian
\begin{equation} \label{eq:pairham}
H^{\Lambda_L} = H_h^{\Lambda_L} + \sum_{z_1, z_2 \in \Lambda_L} V_{\{z_1, z_2 \}} \, ,
\end{equation}
with $V_{\{z_1, z_2 \}}( \mathrm{x}) = F_{\mu} \left( d(z_1, z_2) \right) \cdot V(q_{z_1}, q_{z_2})$.
As one can easily check, the basic assumptions i) - iv) follow if 
$V$ has well-defined, second order partial derivatives and there exist
numbers $C_1$, $\tilde{C}_1$, and $C_2$ such that
\begin{equation}
\max_{i = 1,2} \left| \partial_i V(x,y) \right| \leq C_1 \left( |x| + |y| + \tilde{C}_1 \right)
\end{equation}
and 
\begin{equation}
\max_{i,j \in \{1,2 \}} \left| \partial_i \partial_j V(x,y) \right| \leq C_2 \, .
\end{equation}

If both first order partial derivatives of $V$ are integrable and satisfy the Fourier inversion
formula, then the condition vii) is satisfied when
\begin{equation}
\int_{\mathbb{R}} \int_{\mathbb{R}} \left( |x| + |y| \right) \left( \left| \widehat{\partial_1 V}(x,y) \right| + \left| \widehat{\partial_2 V}(x,y) \right|  \right) \, \rd x \, \rd y \, < \, \infty \, .
\end{equation}

Thus, under the above conditions, the model described by (\ref{eq:pairham}) satisfies the
assumptions of Theorem~\ref{thm:anharmms}, and hence the corresponding locality result (\ref{eq:anharmbdms}) is valid.
\end{example}

\noindent{\it{Proof of Theorem~\ref{thm:anharmms}. }}  
Much of the argument in the proof of Theorem~\ref{thm:anharm2} also applies here.
Again, we fix $L$, regard both Hamiltonians on the same volume, 
drop the dependence of each of the dynamics on $L$, and interpolate.
Let $t>0$ and set 
\begin{equation}\label{eq:phi}
\Phi_t(s) = \{ \alpha_s(\alpha^h_{t-s} (W(f))), W(g) \} \, 
\end{equation}
for $0 \leq s \leq t$. The calculation
\begin{eqnarray}
\frac{\rd}{\rd s} \alpha_s(\alpha^h_{t-s} (W(f))) & = & \sum_{Z \subset \Lambda_L}  \alpha_s \left( \left\{ \alpha^h_{t-s} (W(f)), V_Z  \right\} \right) \nonumber \\
& = &  - i \sum_{Z \subset \Lambda_L} \sum_{z \in Z} {\rm Im} \left[ f_{t-s}(z) \right]  \alpha_s \left(W(f_{t-s}) \right) \cdot \alpha_s \left( \partial_z V_Z  \right) \, , 
\end{eqnarray}
follows readily, and therefore, we derive a differential equation analogous to (\ref{eq:dpsi}); namely
\begin{eqnarray} \label{eq:dphi}
\frac{\rd}{\rd s} \Phi_t(s)  =  i \tilde{\mathcal{L}}_t(s) \Phi_t(s) \, + \, \tilde{\mathcal{Q}}_t(s),  
\end{eqnarray}
where
\begin{equation}
\begin{split}
\tilde{\mathcal{L}}_t(s) & =   - \sum_{Z \subset \Lambda_L} \sum_{z \in Z} {\rm Im} \left[ f_{t-s}(z) \right] \alpha_s( \partial_z V_Z ) \, ,  \\
\tilde{\mathcal{Q}}_t(s) & =  -  i \sum_{Z \subset \Lambda_L}  \sum_{z \in Z} {\rm Im} \left[ f_{t-s}(z) \right] \alpha_s \left( \alpha_{t-s}^h(W(f)) \right) \left\{ \alpha_s ( \partial_z V_Z ) , W(g) \right\}\, .
\end{split}
\end{equation}
Arguing as before, we arrive at the bound
\begin{eqnarray} \label{eq:bdms1}
\left\| \left\{ \alpha_t(W(f)), W(g) \right\} \right\|_{\infty} & \leq & \left\| \left\{ \alpha_t^h(W(f)), W(g) \right\} \right\|_{\infty}   \\
& \mbox{ } & + \sum_{Z \subset \Lambda_L} \sum_{z \in Z} \int_0^t \, \left| {\rm Im} \left[ f_{t-s}(z) \right] \right|  \left\| \left\{ \alpha_s ( \partial_z V_Z ) , W(g) \right\} \right\|_{\infty} \, \rd s \, . \nonumber
\end{eqnarray}

Inserting (\ref{eq:inversion}) into (\ref{eq:bdms1}) and using the harmonic bounds from Corollary~\ref{cor:weylharmbd} with $\mu = \mu_3$, we find that
\begin{eqnarray}\label{eq:msbdtoit}
\begin{split}
& \left\| \left\{ \alpha_t(W(f)), W(g) \right\} \right\|_{\infty}   \leq  C \, \| f \|_{\infty} \, \| g \|_{\infty}
\, e^{ \hat{v} t} \, \sum_{x \in X, y \in Y} F_{\mu_3}(d(x,y))  
\\
 & \quad  + C \, \|f\|_{\infty} \sum_{ Z \subset \Lambda_L}  \sum_{z \in Z} \sum_{x \in X}  F_{\mu_3}(d(x,z)) \int_0^t  \, e^{\hat{v}(t-s)} \, \int_{\mathbb{R}^{Z}} | \widehat{ \partial_z V}(r ; Z) 
| \big\| \left\{ \alpha_s \left( W(r \cdot \delta_Z) \right) , W(g) \right\} \big\|_{\infty} \, \rd r \, \rd s \, , 
\end{split}
\end{eqnarray}
with $C$ as in (\ref{eq:fc}), and we have set $\hat{v} = (\mu_3+ \epsilon)v_{\rm h}( \mu_3 + \epsilon)$ for
notational convenience.

After iterating (\ref{eq:msbdtoit}) $m \geq 1$ times, we find that
\begin{equation}
\left\| \left\{ \alpha_t(W(f)), W(g) \right\} \right\|_{\infty}  \leq C \, \| f \|_{\infty} \, \| g \|_{\infty}
\, e^{\hat{v} t} \, \sum_{x \in X, y \in Y} \sum_{n=0}^m \tilde{a}_n(x,y;t) \, + \, \tilde{R}_{m+1}(t), 
\end{equation}
where 
\begin{equation} \label{eq:ta0}
\tilde{a}_0(x,y;t) = F_{\mu_3}(d(x,y)) \, ,
\end{equation}
\begin{eqnarray}  \label{eq:ta1}
\tilde{a}_1(x,y;t)  & = &   C t \sum_{Z \subset \Lambda_L} \sum_{z_1, z_2 \in Z}  \int_{\mathbb{R}^{Z}} \| r \cdot \delta_Z \|_{\infty}  \cdot |\widehat{ \partial_{z_1} V}(r ; Z)| \, dr \, F_{\mu_3} \left( d(x,z_1) \right) F_{\mu_3} \left( d(z_2, y) \right) \nonumber \\
& \leq &  C t  \sum_{z_1, z_2 \in \Lambda_L}   F_{\mu_3} \left( d(x,z_1) \right) F_{\mu_3} \left( d(z_2, y) \right) \sum_{\substack{Z \subset \Lambda_L\\z_1,z_2 \in Z}} \int_{\mathbb{R}^{Z}}  |r| \cdot | \widehat{ \nabla V}(r; Z)| \, \rd r \, \nonumber \\
& \leq & C \, C_3  \, t  \sum_{z_1, z_2 \in \Lambda_L}   F_{\mu_3} \left( d(x,z_1) \right)  \, F_{\mu_3} \left( d(z_1,z_2) \right) \, F_{\mu_3} \left( d(z_2, y) \right)  \nonumber \\
& \leq & C \, C_3 \, C_{\nu}^2 t  F_{\mu_3} \left( d(x,y) \right) \, ,
\end{eqnarray}
and in general,
\begin{equation}  \label{eq:tan}
\begin{split}
&\tilde{a}_n(x,y;t)   =    \frac{(C t)^n}{n!} \sum_{Z_1, Z_2, \cdots, Z_n \subset \Lambda_L} \sum_{z_{1,1}, z_{1,2} \in Z_1} \sum_{z_{2,1}, z_{2,2} \in Z_2} \cdots \sum_{z_{n,1}, z_{n,2} \in Z_n}  \left( \prod_{j=1}^n \int_{\mathbb{R}^{Z_j}} \| r_j \cdot \delta_{Z_j} \|_{\infty}  \left| \widehat{ \partial_{z_{j,1}} V}(r_j; Z_j) \right| \, \rd r_j \right) 
\\
& \quad \times  F_{\mu_3} \left( d(x,z_{1,1}) \right) \cdot F_{\mu_3} \left( d(z_{1,2}, z_{2,1})\right) \cdots F_{\mu_3} \left( d(z_{n,2}, y) \right) 
\\
& \leq   \frac{(C t)^n}{n!}  \sum_{z_{1,1}, z_{1,2} \in \Lambda_L} \sum_{z_{2,1}, z_{2,2} \in \Lambda_L} \cdots \sum_{z_{n,1}, z_{n,2} \in \Lambda_L}  F_{\mu_3} \left( d(x,z_{1,1}) \right) \cdot F_{\mu_3} \left( d(z_{1,2}, z_{2,1})\right) \cdots F_{\mu_3} \left( d(z_{n,2}, y) \right) 
\\
&  \quad \times \prod_{j=1}^n \sum_{\substack{Z_j \subset \Lambda_L\\z_{j,1},z_{j,2} \in Z_j}} \int_{\mathbb{R}^{Z_j}}  |r_j| \cdot | \widehat{ \nabla V}(r_j; Z_j)| \, \rd r_j \, 
\\
& \leq    \frac{(C \, C_3 \,  t)^n}{n!}  \sum_{z_{1,1}, z_{1,2}, z_{2,1}, z_{2,2}, \cdots, z_{n,1}, z_{n,2} \in \Lambda_L}  F_{\mu_3} \left( d(x,z_{1,1}) \right) \cdot F_{\mu_3} \left( d(z_{1,1}, z_{1,2})\right) \cdot F_{\mu_3} \left( d(z_{1,2}, z_{2,1})\right) \cdots F_{\mu_3} \left( d(z_{n,2}, y) \right) 
\\
& \leq  \frac{(C \, C_3 \, C_{\nu}^2 \, t)^n}{n!}   F_{\mu_3} \left( d(x,y) \right) \, ,
\end{split}
\end{equation}
for any $1 \leq n \leq m$. As before, with $t>0$ fixed, the remainder term $\tilde{R}_{m+1}(t)$ converges to
zero as $m \to \infty$. Thus we have proven (\ref{eq:anharmbdms}) as claimed.
\hfill \qed
  
%
%
%
%

\setcounter{equation}{0}

\section{A Priori Solution Estimates} \label{sec:app}

In this section, we will prove a variety of a priori estimates which will be 
useful in our proofs of the main results.  The underlying argument which
facilitates most of the lemmas below is the well-known Gronwall inequality.
We state and prove a version of this estimate which is tailored to
the present work. A more general bound of this type appears, e.g. in \cite{Agarwal07}. 

\begin{lemma}[Gronwall Inequality] \label{lem:gron}
Let $u: \mathbb{R} \to \mathbb{C}$ satisfy
\begin{equation}\label{eq:ineq}
|u(t)| \leq \alpha(t) + \int_a^t f(t, s)\, | u(s)| \, \rd s
\end{equation}
for all t in [a,b]. If $\alpha$ is non-negative and non-decreasing and $f$ is non-negative and 
continuous with $f( \cdot , s)$ nondecreasing for each fixed $s \in [a,b]$, then
\begin{equation}\label{eq:bdd}
|u(t)| \leq \alpha(t) \exp \left(  \int_a^t f(t,s)\, \rd s\right) \, .
\end{equation}
for all $t$ in $[a,b]$.
\end{lemma}
\begin{proof}
We prove (\ref{eq:bdd}) pointwise. Let $t_0 \in [a,b]$ and observe that 
\begin{equation}\label{eq:ineqt0}
|u(t)| \leq \alpha(t_0) + \int_a^t f(t_0, s) \, |u(s)| \, \rd s \, ,
\end{equation}
holds for all $t \in [a,t_0]$.  Define
\begin{equation}
m(t) = \alpha(t_0) + \int_a^{t} f(t_0, s) |u(s)| \, \rd s \, .
\end{equation}
Clearly, $|u(t)| \leq m(t)$ and the bound
\begin{equation} \label{eq:dm}
m'(t) = f(t_0, t)\, |u(t)| \leq f(t_0, t) m(t) \, ,
\end{equation}
readliy implies
\begin{equation}
|u(t)| \leq m(t) \leq m(0) \, {\rm exp} \left( \int_a^t f(t_0, s) \, \rd s \, \right) \, ,
\end{equation}
for any $t \in [a, t_0]$. Taking $t=t_0$, we have proven (\ref{eq:bdd}).  
\end{proof}

The applications we have in mind concern bounding the solutions of our Hamiltonian flows.
Recall that our general, finite volume, multi-site Hamiltonian, $H^{\Lambda_L} : \mathcal{X}_{\Lambda_L} \to \mathbb{R}$, 
is of the form
\begin{equation} \label{eq:defham2}
H^{\Lambda_L} = H_h^{\Lambda_L} + \sum_{Z \subset \Lambda_L} V_Z \, ,
\end{equation}
and we need a variety of assumptions on the perturbations $V_Z$ to prove our estimates.

We begin with a basic proof of boundedness for the flow $\Phi_t : \mathcal{X}_{\Lambda_L} \to \mathcal{X}_{\Lambda_L}$
corresponding to (\ref{eq:defham2}). As is demonstrated in \cite{LLL}, boundedness follows if the perturbation is
dominated by the harmonic part. For the sake of completeness, we include this argument here.

We assume the perturbation in (\ref{eq:defham2}) above satisfies: 
There exist numbers $C_1 \geq 0$, $\tilde{C_1} \geq 0$, and $\mu_1 \geq 0$ such that
\begin{equation} \label{eq:harmdom}
\left( \sum_{Z \subset \Lambda_L}  \left| \partial_x V_Z ( \mathrm{x}) \right| \right)^2 \leq C_1 \sum_{y \in \Lambda_L} (q_y^2 + \tilde{C_1}) F_{\mu_1} \left(d(x,y) \right) \, ,
\end{equation}
for each $x \in \Lambda_L$ and any $\mathrm{x} \in \mathcal{X}_{\Lambda_L}$. 

\begin{lemma} \label{lem:solbd}
Fix $L \geq 1$ and let $\Phi_t$ denote the flow corresponding to the Hamiltonian $H^{\Lambda_L}$ defined in (\ref{eq:defham2}) above.
If the perturbation satisfies (\ref{eq:harmdom}) described above, then for any $\mathrm{x} \in \mathcal{X}_{\Lambda_L}$ the components of
the flow $\Phi_t( \mathrm{x}) = \{ (q_x(t), p_x(t) ) \}_{x \in \Lambda_L}$ satisfy
\begin{equation} \label{eq:solbd}
\sup_{x \in \Lambda_L} \max \left( |q_x(t) |, \, |p_x(t)| \, \right)  \leq K_1 \, \exp( K_2 \, t),
\end{equation}
where
\begin{equation}\label{eq:K1}
K_1 = K_1( \mathrm{x}) = \sqrt{ \sup_{x \in \Lambda_L} \left( p_x^2(0) + q_x^2(0) + \tilde{C_1} \right)} \, ,
\end{equation}
and
\begin{equation}
K_2 \, = \,  \left| \omega^2 + 2 \sum_{j=1}^{\nu} \lambda_j - 1 \right| + 4 \sum_{j=1}^{\nu} \lambda_j  +  \frac{1}{2} + \frac{C_1}{2} \sum_{x \in \Lambda_L} F_{\mu_1} \left( d(0,x) \right) \, .
\end{equation}
\end{lemma}

\begin{proof} Fix $L \geq 1$, take $x \in \Lambda_L$, and choose $\mathrm{x} \in \mathcal{X}_{\Lambda_L}$. 
Consider the function defined by setting
\begin{equation}
E_x(t) = p_x^2(t) + q_x^2(t) +  \tilde{C_1} \,,
\end{equation}
where $\tilde{C_1}>0$ is the number appearing in (\ref{eq:harmdom}).
From Hamilton's equations, we have that
\begin{eqnarray}
\dot{E}_x(t) & = & 2 p_x(t) \dot{p}_x(t) \, + \, 2 q_x(t) \dot{q}_x(t)  \nonumber \\
& = & - 4 \left( \omega^2 + 2 \sum_{j=1}^{\nu}\lambda_j  -1 \right) \, p_x(t) \, q_x(t) \, + \, 4 p_x(t) \,  \sum_{j=1}^{\nu} \lambda_j \left( q_{x+e_j}(t) \, + \, q_{x-e_j}(t) \right) \nonumber \\
& \mbox{ } & \quad - 2 p_x(t) \sum_{Z \subset \Lambda_L} \left[ \partial_x V_Z \right] \left( \Phi_t( \mathrm{x}) \right) \, ,
\end{eqnarray}
and therefore,
\begin{equation} \label{eq:ebd}
\left| \dot{E}_x(t) \right| \, \leq \, \sum_{y \in \Lambda_L} A_{x,y} E_y(t) \, = \, \left( AE(t) \right)_x
\end{equation}
where the $| \Lambda_L| \times | \Lambda_L|$ matrix $A = (A_{x,y})$ is given by
\begin{equation}
A_{x,y} =  \left\{ \begin{array}{cc} 2 \left| \omega^2 + 2 \sum_{j=1}^{\nu} \lambda_j - 1 \right| + 4 \sum_{j=1}^{\nu} \lambda_j  +1 + C_1 \, F_{\mu_1}(0) & \mbox{if } y=x \, , \\
2 \lambda_j + C_1 \, F_{\mu_1}(1) & \mbox{if } y = x \pm e_j \, , \\
C_1 \, F_{\mu_1} \left( d(x,y) \right) & \mbox{otherwise,}  \end{array} \right.
\end{equation}
and $(AE(t))_x$ is the $x$-th component of this vector.
Denote by $E$ the vector-valued function whose components are $E_x$ and
equip $\mathbb{R}^{|\Lambda_L|}$ with the sup-norm $\| \cdot \|_{\infty}$.
With (\ref{eq:ebd}) it is clear that
\begin{equation}
\left\| \dot{E}(t) \right\|_{\infty} \, \leq \, \left\| AE(t) \right\|_{\infty} \, ,
\end{equation}
and therefore,
\begin{equation}
\left\| E(t) \right\|_{\infty} \leq \left\| E(0) \right\|_{\infty} \, + \, \int_0^t \left\| \dot{E}(s)  \right\|_{\infty} \, \rd s \leq \left\| E(0) \right\|_{\infty} \, + \, \int_0^t \left\| AE(s)  \right\|_{\infty} \, \rd s \, .
\end{equation}
Letting $u(t) = \| E(t) \|_{\infty}$, Lemma~\ref{lem:gron} implies
\begin{equation}
\max \left( |q_x(t)|^2, |p_x(t)|^2 \right) \leq \left\| E(t) \right\|_{\infty} \, \leq \, \left\| E(0) \right\|_{\infty} \, \exp \left( \|A\|_{\infty} \, t \right),
\end{equation}
from which (\ref{eq:solbd}) is clear.
\end{proof}

As will become clear in the proof of Lemma~\ref{lem:pbbd} below, the main quantities of interest for us are the
derivatives of the components of the flow with respect to the initial conditions. The next lemma provides
explicit estimates on these functions. To prove it we need the following additional assumption on our perturbation.
Assume there exists constants $C_2 \geq 0$ and $\mu_2 \geq 0$ for which: given any $L \geq 1$ and any
pair $x,y \in \Lambda_L$, the bound
\begin{equation} \label{eq:2derbd}
\sum_{Z \subset \Lambda_L} \left| \left[ \partial_x \partial_y V_Z \right] ( \mathrm{x}) \right| \, \leq \, C_2 F_{\mu_2} \left( d(x,y) \right) \, ,
\end{equation}
holds for all points $\mathrm{x} \in \mathcal{X}_{\Lambda_L}$.

\begin{lemma} \label{lem:dsolbd}
Fix $L \geq 1$ and let $\Phi_t$ denote the flow corresponding to the Hamiltonian $H^{\Lambda_L}$ defined in (\ref{eq:defham2}) above.
If the perturbation satisfies (\ref{eq:harmdom}) and (\ref{eq:2derbd}), then for any $\mathrm{x} \in \mathcal{X}_{\Lambda_L}$ the components of
the flow $\Phi_t( \mathrm{x}) = \{ (q_x(t), p_x(t) ) \}_{x \in \Lambda_L}$ satisfy
\begin{equation} \label{eq:dsolbdq}
\sup_{x,y \in \Lambda_L} \max \left( \left| \frac{\partial q_x(t)}{ \partial q_y(0)} \right|, \, \left| \frac{\partial q_x(t)}{ \partial p_y(0)} \right| \, \right)  \leq \max(1,2t) \, \exp( K \, t^2),
\end{equation}
and
\begin{equation} \label{eq:dsolbdp}
\sup_{x,y \in \Lambda_L} \max \left(\left| \frac{\partial p_x(t)}{ \partial q_y(0)} \right|, \, \left| \frac{\partial p_x(t)}{ \partial p_y(0)} \right| \, \right)  \leq 1 + t \left(K + 2 \sum_{j=1}^{\nu} \lambda_j \right) \, \max(1,2t) \, \exp( K \, t^2),
\end{equation}
where
\begin{equation}
K = 2 \omega^2 + 8 \sum_{j=1}^{\nu} \lambda_j + C_2 \sum_{x \in \Lambda_L} F_{\mu_2} \left( d(0,x) \right) \, .
\end{equation}
\end{lemma}

\begin{proof}
We begin with a proof of (\ref{eq:dsolbdq}).
Using Hamilton's equations, we find that
\begin{eqnarray} \label{eq:qxt}
q_x(t) \, - \, q_x(0) \, - \,  2 \, t \, p_x(0) \,  & = & 2 \, \int_0^t \int_0^s \dot{p}_x(r) \, \rd r \, \rd s    \\
& = &  - \,  4 \left( \omega^2 \, + \, 2 \sum_{j=1}^{\nu} \lambda_j \right) \int_0^t (t-s) \, q_x(s) \, \rd s \, \nonumber \\
& \mbox{ } & \quad  + \, 4 \, \sum_{j=1}^{\nu} \lambda_j \int_0^t (t-s) \left(q_{x + e_j}(s) + q_{x- e_j}(s) \right) \, \rd s \nonumber \\
& \mbox{ } & \quad \quad - \, 2 \,  \sum_{Z \subset \Lambda_L} \int_0^t (t-s) \left[ \partial_x V_Z \right]  \left( \Phi_s( \mathrm{x}) \right)  \, \rd s  \nonumber
\end{eqnarray}
and therefore
\begin{eqnarray}\label{eq:dqxq0}
\frac{\partial q_x(t)}{\partial q_y(0)}  & = &  \delta_x(y)  - \,  4 \left( \omega^2 \, + \, 2 \sum_{j=1}^{\nu} \lambda_j \right) \int_0^t (t-s) \, \frac{ \partial q_x(s)}{ \partial q_y(0)} \, \rd s  
\\
& \mbox{ } & \quad  + \, 4 \, \sum_{j=1}^{\nu} \lambda_j \int_0^t (t-s) \left( \frac{ \partial q_{x + e_j}(s)}{ \partial q_y(0)} \, + \, \frac{ \partial q_{x- e_j}(s)}{ \partial q_y(0)} \right) \, \rd s \nonumber 
\\
& \mbox{ } & \quad \quad - \, 2 \, \sum_{z \in \Lambda_L} \sum_{Z \subset \Lambda_L}  \int_0^t (t-s) \left[ \partial_z \partial_x V_Z \right]  \left( \Phi_s( \mathrm{x}) \right) \cdot \frac{ \partial q_z(s)}{ \partial q_y(0)}  \, \rd s  \nonumber \, .
\end{eqnarray}
More succinctly, we have found that
\begin{equation} \label{eq:matbdq}
\left| \frac{\partial q_x(t)}{\partial q_y(0)} \right|   \leq  \delta_x(y) + \sum_{z \in \Lambda_L} \int_0^t (t-s) A_{x,z} \left| \frac{ \partial q_z(s)}{ \partial q_y(0)} \right| \, \rd s
\end{equation}
where the $| \Lambda_L| \times | \Lambda_L|$ matrix $A = (A_{x,z})$ has entries
\begin{equation}
A_{x,z} =  \left\{ \begin{array}{cc} 4 \left(\omega^2 + 2 \sum_{j=1}^{\nu} \lambda_j \right) +2 C_2 F_{\mu_2}(0) & \mbox{if } z=x \, , \\
4 \lambda_j \, + \, 2 C_2 F_{\mu_2}(1) & \mbox{if } z = x \pm e_j \, , \\
2 C_2 F_{\mu_2} \left( d(x,z) \right)  & \mbox{otherwise.}  \end{array} \right.
\end{equation}

If, for each fixed $y \in \Lambda_L$, we denote by $\partial_y q : \mathbb{R} \to \mathbb{R}^{|\Lambda_L|}$ the vector-valued function
whose components are given by $\left| \frac{ \partial q_x(t)}{ \partial q_y(0)} \right|$, then (\ref{eq:matbdq}) implies 
\begin{equation}
\left\| \partial_y q(t) \right\|_{\infty}  \leq 1 +  \int_0^t (t-s) \left\| A \partial_y q(s) \right\|_{\infty} \, ds \, .
\end{equation}
Again letting $u(t) = \| \partial_y q(t) \|_{\infty}$, Lemma~\ref{lem:gron} yields the estimate
\begin{equation}
\left\| \partial_y q(t) \right\|_{\infty}  \leq \exp \left( \frac{\| A \|_{\infty} \, t^2}{2} \right) \, .
\end{equation}

Quite similarly, the bound
\begin{equation} \label{eq:matbdp}
\left| \frac{\partial q_x(t)}{\partial p_y(0)} \right|   \leq  2 t \, \delta_x(y) + \sum_{z \in \Lambda_L} \int_0^t (t-s) A_{x,z} \left| \frac{ \partial q_z(s)}{ \partial q_y(0)} \right| \, \rd s \, ,
\end{equation}
follows from (\ref{eq:qxt}) with the same matrix $A$. This proves (\ref{eq:dsolbdq}).

The bound for $p_x(t)$ follows from (\ref{eq:dsolbdq}). In fact, it is easy to see that
\begin{eqnarray} \label{eq:pxt}
p_x(t) & = & p_x(0)  - 2 \left( \omega^2 \, + \, 2 \sum_{j=1}^{\nu} \lambda_j \right) \int_0^t q_x(s) \, \rd s  \\
& \mbox{ } & \quad + \, 2 \, \sum_{j=1}^{\nu} \lambda_j \int_0^t \left(q_{x + e_j}(s) + q_{x- e_j}(s) \right) \, \rd s \,
  -  \,  \sum_{Z \subset \Lambda_L} \int_0^t \left[ \partial_x V_Z \right]  \left( \Phi_s( \mathrm{x}) \right)  \, \rd s \, . \nonumber   
\end{eqnarray}
Using (\ref{eq:dsolbdq}), (\ref{eq:dsolbdp}) readily follows. 
\end{proof}

\begin{lemma} \label{lem:pbbd} Let $X, Y \subset \mathbb{Z}^{\nu}$ be finite sets and take $L_0 \geq 1$ large enough
so that $X,Y \subset \Lambda_{L_0}$. For any $L \geq L_0$ and $t \in \mathbb{R}$, denote by $\alpha_t^L$ the dynamics corresponding
to (\ref{eq:defham2}). If the perturbation satisfies (\ref{eq:harmdom}) and (\ref{eq:2derbd}), then there exist positive numbers $K_1$ and $K_2$, both
independent of $L$, for which: given any functions $f: X \to \mathbb{C}$ and $g : Y \to \mathbb{C}$, the bound
\begin{equation} \label{eq:pbbd}
\left\| \left\{ \alpha_t^L \left( W(f) \right), W(g)  \right\} \right\|_{\infty} \, \leq \, K_1 \, |X| \, |Y| \, \| f \|_{\infty} \, \| g \|_{\infty} \, \exp( K_2 \, t^2 ) \, , 
\end{equation}
holds for all $t \in \mathbb{R}$.
\end{lemma}

\begin{proof}
We first fix $L \geq L_0$ as in the statement of the lemma and prove the estimate on $\Lambda_L$. In this case, we
suppress the dependence of most quantities on $L$ to ease notation. Now, recall that for any fixed point $\mathrm{x}$,
\begin{equation} \label{eq:basepb}
\begin{split}
\left[ \left\{  \alpha_t(W(f)), W(g) \right\} \right]( \mathrm{x})  =  \sum_{y \in \Lambda_L} & \frac{ \partial}{ \partial q_y} \left[ \alpha_t(W(f)) \right] ( \mathrm{x}) \cdot  \frac{ \partial}{ \partial p_y} \left[ W(g) \right]( \mathrm{x})  \\
& -  \sum_{y \in \Lambda_L}  \frac{ \partial}{ \partial p_y} \left[ \alpha_t(W(f)) \right] ( \mathrm{x}) \cdot  \frac{ \partial}{ \partial q_y} \left[ W(g) \right]( \mathrm{x}) \, .
\end{split}
\end{equation}
Since
\begin{equation}
 \frac{ \partial}{ \partial p_y} \left[ W(g) \right]( \mathrm{x}) = i {\rm Im} \left[ g(y) \right] \, \left[ W(g) \right] ( \mathrm{x}) \quad \mbox{and} \quad  \frac{ \partial}{ \partial q_y} \left[ W(g) \right]( \mathrm{x}) = i {\rm Re} \left[ g(y) \right] \, \left[ W(g) \right] ( \mathrm{x}) \, ,
\end{equation}
the sums in (\ref{eq:basepb}) above are only over those $y$ in the support of $g$. The derivative of the time-evolved quantities may also be calculated, e.g.,
\begin{equation} \label{eq:dqyat}
\frac{ \partial}{ \partial q_y} \left[ \alpha_t(W(f)) \right] ( \mathrm{x}) \, = \, i \left( \sum_{x \in \Lambda_L} {\rm Re} \left[ f(x) \right] \cdot \frac{ \partial q_x(t)}{\partial q_y} \, + \, {\rm Im} \left[ f(x) \right] \cdot \frac{ \partial p_x(t)}{\partial q_y} \right) \cdot \left[ \alpha_t(W(f)) \right] ( \mathrm{x}) \, .
\end{equation}
Clearly, the sum in (\ref{eq:dqyat}) is only over those $x$ in the support of $f$, and a similar formula holds for
the derivative with respect to $p_y$.
Thus,
\begin{equation*}
\left| \left[ \left\{  \alpha_t(W(f)), W(g) \right\} \right]( \mathrm{x})  \right| \, \leq \, 4\, \| f \|_{\infty} \, \| g \|_{\infty} 
\sum_{x \in X, y \in Y} \max \left( \left| \frac{ \partial q_x(t)}{\partial q_y} \right| , \left| \frac{ \partial p_x(t)}{\partial q_y} \right| , \left| \frac{ \partial q_x(t)}{\partial p_y} \right| , \left| \frac{ \partial p_x(t)}{\partial p_y} \right| \right) \, .
\end{equation*} 
Using Lemma~\ref{lem:dsolbd}, (\ref{eq:pbbd}) immediately follows.
\end{proof}

{\it Acknowledgements.} 

Both authors would like to thank Bruno Nachtergaele and Benjamin Schlein for their  insight and useful discussions. The authors would like to acknowledge the hospitality of the Erwin Schrodinger Institute, and specifically the organizers of the Summer School on Current topics in Mathematical Physics, where work on this project began. This article is based on work supported in part by the U.S. National Science Foundation. In particular, R.S. was supported under grant \# DMS-0757424, and H.R. received support from NSF grants \# DMS-0605342, \# DMS-07-57581 and VIGRE grant \# DMS-0636297.

\end{document}